\documentclass[12pt]{article}
\pdfoutput=1
\def\citep#1{\cite{#1}}

\newif\ifpreprint
\preprinttrue

\usepackage{fullpage}

\usepackage[utf8]{inputenc} 
\usepackage[T1]{fontenc}    

\usepackage{amsfonts}       
\usepackage{amssymb}
\usepackage{amsmath}
\usepackage{euscript}
\usepackage{stmaryrd} 
\usepackage{color}

\usepackage[colors]{optsys}

\newcommand{\GenericSeparation}[1]%
        {\mathrel{\raisebox{-0.2em}{$\substack{\text{\small $\Vert$} \\[-1.7mm] \line(1,0){10}\\ #1}$}}}

\newcommand{\ConditionalTopologicalSeparation}{\GenericSeparation{t}}

\newcommand{\Tsep}{\GenericSeparation{t}}

\newcommand{\LastElement}[1]{{#1}^\star}
\newcommand{\FirstElements}[1]{{#1}^-}

\renewcommand{\AGENT}{{\mathbb A}}
\renewcommand{\Agent}{{A}}

\newcommand{\cut}{\psi}
\newcommand{\ORDER}{\Sigma}
\newcommand{\ordering}{\varphi}
\newcommand{\totalordering}{\rho}
\newcommand{\numAGENT}{| \AGENT |}

\newcommand{\AgentSubsetW}{W}
\newcommand{\AgentSubsetY}{Y}
\newcommand{\AgentSubsetZ}{Z}
\newcommand{\HistorySubset}{H}

\newcommand{\ParentRelation}{\mathcal{P}} 

\newcommand{\SubWH}{{\AgentSubsetW,\!\HistorySubset}}
\newcommand{\Precedence}{\ParentRelation}
\newcommand{\PrecedenceWH}{\ParentRelation_{\AgentSubsetW,\HistorySubset}}

\newcommand{\PrecedenceEmptyH}{\ParentRelation_{\emptyset,\HistorySubset}}
\newcommand{\PrecedenceEmptyHISTORY}{\ParentRelation_{\emptyset,\HISTORY}}

\usepackage{booktabs}
\usepackage{graphicx}
\usepackage{subcaption}
\usepackage{csquotes}
\usepackage{cancel}

\newcommand{\RandomVariable}{\va{\Control}}

\newcommand{\bm}[1]{#1}
\renewcommand{\textbf}{\emph}

\newcommand{\residual}{E} 

\newcommand{\partialReducedSolutionMap}{\widetilde{\ReducedSolutionMap}}

\newcommand{\ClosedTopologyWH}{{\cal F}^{\text{\tiny $\SubWH$}}}
\newcommand{\TopologicalClosure}[1]{\overline{#1}^{\text{\tiny $\SubWH$}}}

\newcommand{\ConditionalAncestor}{\TransitiveReflexiveClosure{\PrecedenceWH}}

\newcommand{\PROBA}{\PP}
\renewcommand{\range}[1]{\|#1\|}

\usepackage{tikz}
\usetikzlibrary{shapes.geometric,calc,patterns}
\usetikzlibrary{bayesnet}
\usetikzlibrary{positioning}

\usepackage{amsmath}
\graphicspath{{./}}

\title{Causal Inference Theory with\\Information Dependency Models}

\author{Benjamin Heymann\footnote{Criteo AI Lab, Paris, France}, 
  Michel De Lara\footnote{CERMICS, Ecole des Ponts, Marne-la-Vall\'ee, France},
  Jean-Philippe Chancelier\footnotemark[2]}

\date{\today}

\begin{document}

\maketitle

\begin{abstract}
  Inferring the potential consequences  of an unobserved event is a fundamental scientific
  question. To this end,  Pearl's celebrated do-calculus provides a set of
  inference rules to derive an interventional probability from an observational
  one. 
  In this framework, the primitive causal relations are encoded as
  functional dependencies in a Structural Causal Model (SCM),
 which are generally mapped into a
  Directed Acyclic Graph (DAG) in the absence of cycles.
  In this paper, by contrast, we capture causality without reference to  graphs or
  functional dependencies, but with information fields and
  Witsenhausen's intrinsic model.
  The three rules of do-calculus reduce to a
  unique sufficient condition for conditional independence, the
  topological separation, which presents interesting theoretical and practical
  advantages over the d-separation.  
  With this unique rule, we can deal 
  with systems that cannot be
  represented with DAGs, for instance systems with cycles and/or
  `spurious' edges.
  We treat an example that cannot be handled --- to
  the extent of our knowledge --- with the tools of the current
  literature.
  We also explain why, in the presence of cycles, the theory of causal inference 
  might require different tools, depending on  whether the random
  variables are  discrete or  continuous. 
\end{abstract}

\section{Introduction}

As the world shifts toward more and more data-driven decision-making, causal
inference is taking more space in applied sciences, statistics and machine
learning.  This is because it allows for better, more robust decision-making,
and provides a way to interpret the data that goes beyond correlation
\citep{pearl2018book}.  For instance, causal inference provides a language to
describe and solve Simpson's paradox, which embodies the ``correlation is not
causation'' principle as can be found in any ``Statistics~101'' basic course.
The main concern in causal inference is to compute post-intervention probability
distributions from observational data.  For this purpose, graphical models are
practical because they allow representing assumptions easily and benefit from an
extensive scientific literature.

In his seminal work~\citep{pearl1995causal}, Pearl builds on graphical
models~\citep{cowell2006probabilistic} to introduce the so-called do-calculus.
Causal graphical models move the focus from joint probability distributions to
functional dependencies thanks to the Structural Causal Model (SCM) framework.
Several extensions to this do-calculus have been proposed
recently~\citep{winn2012causality,lattimore2019replacing,tikka2019identifying,correa2020a}.
Pearl's seminal paper supposes a Directed Acyclic Graph (DAG) structure.

In this paper, we bring a new, complementary view to the causal reasoning
toolbox by leveraging the concept of information fields  and
  Witsenhausen's intrinsic model.
The framework  we introduce is general,  unifying, and may be used to study  causal inference
in both recursive and nonrecursive systems
\citep{Halpern:2000} (i.e. with and without cycles). It allows for
spurious edges, and simplifies the statement of Pearl's three rules of do-calculus.

DAGs modeling does not rely directly on random variables but on joint
probability distributions (see \cite[footnote~3]{Pearl2011} or
\cite[Appendix~A]{peters2017elements}). By contrast, our approach requires going
back to the classical primitives of probabilistic models: sample sets,
$\sigma$-fields, measurable maps and random variables.  We exploit the generally
overlooked expressiveness of this underlying structure.  The cost for this
conceptual generalization is a bit of abstraction: in what we propose, the
structure is implicit, and there are no arrows.

This paper, however, has been written so that the main messages 
can be understood with the usual graphical concepts used in the field of causal
inference: the notion of topological separation is
explained for the specific case of DAGs;
Theorem~\ref{th:do-calculus3rules} and
Examples~\ref{ex:easy}, \ref{ex:non-rec-easy} and
\ref{ex:w} should be readable without the concept of information
field.
In addition, this paper was written in parallel to two other
papers~\cite{Chancelier-De-Lara-Heymann-2021,De-Lara-Chancelier-Heymann-2021};
the three of them aim at providing another perspective on conditional independence and do-calculus.

\paragraph{Related work and contributions.}

We extend the causal modeling toolbox thanks to two notions: information fields
and topological separation.
These two  notions rely on the foundational work produced by
Witsenhausen in the seventies~\cite{Witsenhausen:1971a}.   
The concept of information field extends the
expressiveness of the Structural Causal Model, and allows for instance to
naturally encode context specific independence~\cite{tikka2019identifying}.  In
the companion papers~\cite{Chancelier-De-Lara-Heymann-2021,De-Lara-Chancelier-Heymann-2021}, we show an
equivalence between Pearl's d-separation and a new notion that we introduce,
the conditional topological separation.  The topological separation is practical
because it just requires to check that two sets are disjoints (see
Examples~\ref{fig:kuh}). By contrast, the d-separation requires to check that
\emph{all} the paths that connect two variables are blocked.  Moreover, as its
name suggests, the topological separation has a theoretical
interpretation. Specifically, the topological separation allows us to go beyond
DAGs and even graphical models.

Our main results are (i) Theorem~\ref{th:do-calculus3rules}, which is a
generalization of do-calculus that can be applied in particular to nonrecursive
systems~\cite{bongers2020foundations} and which subsumes several recent results,
and (ii) Lemma~\ref{th:decomposition} which provides insight into the
machinery behind Theorem~\ref{th:do-calculus3rules}.  We pinpoint the novelty of
our approach with Example~\ref{ex:w}, a system with cycles where our framework
identifies a probabilistic independence that the framework developed in
~\cite{pmlr-v115-forre20a} (for cycles) does not. We explain in
Sect.~\ref{subseq:discrete-continuous} that the differences between the
framework developped in~\cite{pmlr-v115-forre20a} and ours comes from a
fundamental difference whether in the discrete or in the continuous
  setting regarding random variables.
  \medskip

  The paper is organized in two parts as follows.
First, we provide what we think will be of interest for application minded
researchers in Sect.~\ref{Definition_of_Information_Dependency_Models}
and~\ref{sec-Main-results}.
Sect.~\ref{Definition_of_Information_Dependency_Models} introduces the notion of
Information Dependency Model, which is another way of looking at
systems that can be represented with SCMs.
We then presents our main results in Sect.~\ref{sec-Main-results}:
we restate Pearl's do-calculus theorem in terms of topological separation.
Second, we present the theoretical foundation of those results in
Sect.~\ref{Presentation_of_Witsenhausen_product_model_and_solvability}
and~\ref{sec:formal-proof}.
Sect.~\ref{Presentation_of_Witsenhausen_product_model_and_solvability}
present Witsenhausen's intrinsic model upon which we build our contributions.
We provide the proofs in Sect.~\ref{sec:formal-proof}.

\section{Definition of Information Dependency Models}
\label{Definition_of_Information_Dependency_Models}

In~\S\ref{sec:information_fields}, we provide background on
$\sigma$-fields and introduction the  Information Dependency Models.
Then, in~\S\ref{sec:conditional}, we define conditional precedence.

\subsection{Information fields and Information Dependency Models}
\label{sec:information_fields}

We start with a few reminders from measure (and probability) theory.  A
\emph{$\sigma$-field} (henceforth sometimes referred to as \emph{field}) over a
set~$\SET$ is a subset $\tribu{\Set} \subset 2^\SET$, containing~$\SET$, and
which is stable under complementation and under countable union. The couple
\( \np{\SET,\tribu{\Set}} \) is called a \emph{measurable space}. The trivial
$\sigma$-field over the set~$\SET$ is \( \{ \emptyset, \SET \} \). The complete
$\sigma$-field over the set~$\SET$ is \( 2^\SET \). When
\( \tribu{\Set}' \subset \tribu{\Set} \) are two $\sigma$-fields over the
set~$\SET$, we say that $\tribu{\Set}'$ is a \emph{subfield} of~$\tribu{\Set}$.
If  \( \tribu{\Set} \) is a $\sigma$-field over the
set~$\SET$ and if $\SET' \subset \SET$, then \( \tribu{\Set} \cap \SET' =
\nset{ D \cap \SET' }{ D \in \tribu{\Set} } \) is a $\sigma$-field over the
set~$\SET'$, called the \emph{trace subfield} of~$\tribu{\Set}$ over~$\SET'$.
If \( \np{\SET_i,\tribu{\Set}_i} \), $i=1,2$ are two measurable spaces, we
denote by $\tribu{\Set}_1\otimes\tribu{\Set}_2 $ the \emph{product
  $\sigma$-field} on $\SET_1\times\SET_2$ generated by the rectangles
$\nset{D_1\times D_2}{D_i\in \tribu{\Set}_i, i=1,2}$.  More generally, if
\( \sequence{\np{\SET_\scenario,\tribu{\Set}_\scenario}}{\scenario\in\SCENARIO}
\) is a family of measurable spaces, we denote by
\( \bigotimes_{\scenario\in\SCENARIO}\tribu{\Set}_\scenario \) the \emph{product
  $\sigma$-field} on $\prod_{\scenario\in\SCENARIO}\SET_\scenario $ generated by
the cylinders. 
Let $(\Omega, \tribu{\NatureField})$ and $(\CONTROL,\tribu{\Control})$ be two measurable spaces, probability theory defines a
\emph{random variable} as a measurable mapping from $(\Omega,\tribu{\NatureField})$ to
$(\CONTROL,\tribu{\Control})$, that is, a mapping \(\lambda:\Omega\to \CONTROL\) satisfying \(
  \wstrategy^{-1} (\tribu{\Control})\subset \tribu{\NatureField}\). When equipped with a probability, 
$\PROBA$, a measurable space $(\Omega,\tribu{\NatureField})$ is called a probability space and is denoted by 
$(\Omega,\tribu{\NatureField}, \PROBA)$.

\subsubsection{Structural Causal Models (informal definition)}
\label{Structural_Causal_Models_(informal_definition)}

Thus equipped, we now discuss the standard way to model causal hypotheses using
\emph{Structural Causal Models (SCMs})~\cite{peters2017elements}.

Let $\AGENT$ be a set and, for each $\agent \in \AGENT$, a given probability space
$(\Omega_{\agent}, \tribu{\NatureField}_\agent,\PROBA_{a})$. We consider the
product probability space
$(\Omega,  \tribu{\NatureField}, \PROBA)$ where   
\( \Omega = \prod_{\agent\in\AGENT}\Omega_{\agent} \),
$\tribu{\NatureField}=\bigotimes_{\agent \in \AGENT} \tribu{\NatureField}_\agent$, and
$\PROBA = \bigotimes_{\agent \in \AGENT} \PROBA_{\agent}$.
Let \(
\sequence{\np{\CONTROL_{\agent},\tribu{\Control}_{\agent}}}{\agent\in\AGENT} \)
be a family of measurable spaces.

An SCM consists of a family $(\policy_\agent)_{\agent\in\AGENT}$ of mappings
(or \emph{assignments}),
where each~\( \policy_\agent \) has codomain~\( \CONTROL_{\agent} \), 
alongside with a parental mapping  $P:\AGENT\to 2^\AGENT$, 
and of a family of random variables $\nseqp{\RandomVariable_\agent}{\agent \in \AGENT}$,
all defined on the probability space $(\Omega,  \tribu{\NatureField}, \PROBA)$
an such that each~\( \RandomVariable_\agent \) has codomain~\( \CONTROL_{\agent} \), 
with the property that
\begin{equation}
  \label{eq:SCM}
  \RandomVariable_\agent\np{\omega} =
  \policy_\agent\bp{\omega_\agent,\RandomVariable_{P(\agent)}\np{\omega}}
  \eqsepv \forall \omega\in\Omega
  \eqsepv \forall\agent\in\AGENT
  \eqfinv 
\end{equation}
where $\omega_\agent$ is the projection of~$\omega$ on~$\Omega_\agent$.

To get the \textbf{graphical representation} of a SCM --- as a subgraph of the
graph~$(\AGENT,\AGENT\times\AGENT)$ --- we draw an arrow $\agent\to\bgent$
whenever $\agent\in P(\bgent)$. Usually, the graphical representation is assumed
to be a DAG, which means that the parental mapping induces a partial order on
the set $\AGENT$. We will not need this assumption here.  Sufficient
  condition to obtain causal properties, relying only on the 
  graphical representation (which is uniquely defined by the parental mapping~$P$),
  have been developed by many authors. These
  conditions take their importance from the fact that they
  short-circuit reasoning on the assignement mappings.
 For a given applied problem, the SCM is derived
from expert knowledge, assumptions and data analysis methods.  The SCM is a
central tool in causal analysis but its graphical representation does not
naturally account for situations such as Context Specific
Independence (see~\cite{tikka2019identifying}, and
Example~\ref{example:tikka}), where some edges are spurious.

\subsubsection{Information Dependency Models (first informal definition)}

From Equation~\eqref{eq:SCM}, the set of  arguments of the assignement
mapping~$\policy_\agent$ depends on $\agent$ in the formalism of the SCM, 
(remember that $\policy_\agent$ is the assigment function of
$\RandomVariable_\agent$ for some $\agent\in\AGENT$).
By contrast, in the Information Dependency Model formulation, the assignement mappings
have a common domain, that we call the \emph{configuration space},
which is the product space\footnote{Also called \emph{hybrid
    space}~\cite{Witsenhausen:1971a}, hence the $\HISTORY$ notation.}
\( \HISTORY = \produit{\Omega}{ \prod \limits_{\agent \in \AGENT}
  \CONTROL_{\agent}} \).  The \emph{configuration field}
\( \tribu{\History} =\oproduit{\tribu{\NatureField}}{\bigotimes \limits_{\agent \in \AGENT}
  \tribu{\Control}_{\agent}} \) is a $\sigma$-field
over~\( \HISTORY \).
We then extend the definition of SCM thanks to the following observation: we can
express the SCM in~\S\ref{Structural_Causal_Models_(informal_definition)}
by saying that $\policy_\agent$ is a map from $\HISTORY$ to $\CONTROL_{\agent}$
(\( \wstrategy_{\agent} : \np{\HISTORY,\tribu{\History}} \to
\np{\CONTROL_{\agent},\tribu{\Control}_{\agent}}\)) while imposing that
$\policy_\agent$ ``only depends on $\RandomVariable_{P(\agent)}$ and
$\omega_\agent$''. It is standard (see \cite[Chap. 1
p. 18]{Dellacherie-Meyer:1975}) in probability theory that such property is ---
under mild assumptions --- equivalent to a so-called \emph{measurability
  constraint} on the random variable $\RandomVariable_{\agent}$.  Hence, the
informal definition~\eqref{eq:SCM} of a SCM can be restated as
\begin{equation}
  \label{eq:scmReformulation}
  \wstrategy_{\agent}^{-1} (\tribu{\Control}_{\agent})
  \subset \oproduit{ \tribu{\NatureField}_\agent\otimes
    \bigotimes \limits_{\bgent \neq\agent}    \{\emptyset,\Omega_\bgent\}  }
  { \bigotimes \limits_{\bgent \in
      P(\agent)}\tribu{\Control}_{\bgent}\otimes     \bigotimes \limits_{\bgent \not\in
      P(\agent)} \{\emptyset,\CONTROL_{\bgent}\} }
  \eqfinv 
\end{equation}
or, with a slight abuse of notations that we will sometimes use throughout this
presentation\footnote{We omit the trivial fields in the product on the
  right-hand side of Equation~\eqref{eq:scmReformulation2}.}
\begin{equation}
  \label{eq:scmReformulation2}
  \wstrategy_{\agent}^{-1} (\tribu{\Control}_{\agent})
  \subset \oproduct{\tribu{\NatureField}_\agent}%
  { \bigotimes_{\bgent \in P(\agent)}\tribu{\Control}_{\bgent}}
  \eqfinp 
\end{equation}
Informally, an information field is anything one may want to see on the
right-hand side of Equation~\eqref{eq:scmReformulation2}.
For instance, consider the case where \( \AGENT=\na{\agent,\bgent,\cgent} \)
and suppose that all fields contain the singletons.
If \( \wstrategy_{\agent}^{-1} (\tribu{\Control}_{\agent})
\subset \oproduit{ \tribu{\NatureField}_\agent \otimes 
  \{\emptyset,\Omega_\bgent\} \otimes \{\emptyset,\Omega_\cgent\} }
{ \{\emptyset,\CONTROL_\agent\} \otimes \{\emptyset,\CONTROL_\bgent\} \otimes
  \{\emptyset,\CONTROL_\cgent\} } \), that we abusively write 
\( \wstrategy_{\agent}^{-1} (\tribu{\Control}_{\agent})
\subset \tribu{\NatureField}_\agent \), this means that 
\( \wstrategy_{\agent}\Icouple{\control_\agent,\control_\bgent,\control_\cgent}%
{\omega_\agent,\omega_\bgent,\omega_\cgent} 
=
\wstrategy_{\agent}\Icouple{\cancel{\control_\agent},\cancel{\control_\bgent},\cancel{\control_\cgent}}%
{\omega_\agent,\cancel{\omega_\bgent},\cancel{\omega_\cgent}}  \)
only depends on~$\omega_\agent$, that is, only depends on its own ``source of
uncertainty'' (the field $\tribu{\NatureField}_\agent$).
If (abusively) \( \wstrategy_{\bgent}^{-1} (\tribu{\Control}_{\bgent})
\subset \oproduit{ \tribu{\NatureField}_\agent }{ \tribu{\Control}_{\cgent} } \), this means that 
\( \wstrategy_{\bgent}\Icouple{\control_\agent,\control_\bgent,\control_\cgent}%
{\omega_\agent,\omega_\bgent,\omega_\cgent} 
=
\wstrategy_{\bgent}\Icouple{\cancel{\control_\agent},\cancel{\control_\bgent},\control_\cgent}%
{\omega_\agent,\cancel{\omega_\bgent},\cancel{\omega_\cgent}}  \)
only depends on~$\Icouple{\control_\cgent}{\omega_\agent}$, that is, 
only depends on the uncertainty~$\omega_\agent$ 
(the field $\tribu{\NatureField}_\agent$)
and on the variable~$\control_\cgent$
(the field $\tribu{\Control}_{\cgent}$).
%
%
More complex examples will be given later.

After having discussed how SCMs can be interpreted with the help of information
fields, we propose the name Information Dependency Model for their extension.

\begin{definition}[Information Dependency Model]
  An \emph{Information} \emph{Dependency} \emph{Model} (IDM) is a collection
  $(\tribu{\Information}_{\agent})_{\agent\in\AGENT}$ of subfields of
  $\tribu{\History}$ such that, for any $\agent \in \AGENT$,
  \( \tribu{\Information}_{\agent} \subset \oproduit{\tribu{\NatureField}_\agent}{\bigotimes \limits_{\bgent
      \in \AGENT} \tribu{\Control}_{\bgent}} \).
  The subfield $\tribu{\Information}_{\agent}$ is called the \textbf{information
    field} of $\agent$.
\label{de:Information_Dependency_Model}
\end{definition}

The SCM defining property~\eqref{eq:SCM} is now expressed in term of the \textbf{measurability property}
\begin{equation}
  \label{eq:decision_rule}
  \wstrategy_{\agent}^{-1} (\tribu{\Control}_{\agent})
  \subset \tribu{\Information}_{\agent} \eqsepv \forall \agent\in\AGENT
  \eqfinp
\end{equation}
Property~\eqref{eq:decision_rule} expresses, in a very general way, that the
random variable~$\RandomVariable_\agent$ may only depend upon the available
information~$\tribu{\Information}_{\agent}$.  It is a generalization of the
notion of nonanticipativity constraint, or of adapted process with respect to a
filtration, in stochastic control.  For a given applied problem, like for the
SCM, the IDM can be derived from expert knowledge, assumptions and data analysis
methods.
\begin{remark}
\label{remark:scmtoidm}
Any SCM can be mapped into an IDM
as we obtain from Equation~\eqref{eq:SCM} that, for all $\agent\in\AGENT$, we have that 
  \(
  \wstrategy_{\agent}^{-1} (\tribu{\Control}_{\agent})
  \subset \tribu{\Information}_{\agent}\) with 
  \(
  \tribu{\Information}_{\agent} =
  \oproduit{\tribu{\NatureField}_\agent}{\bigotimes_{\bgent \in P(\agent)} \tribu{\Control}_{\bgent}} 
  \subset \oproduit{\tribu{\NatureField}_\agent}{\bigotimes \limits_{\bgent
      \in \AGENT} \tribu{\Control}_{\bgent}} \).
\end{remark}

\begin{figure}[h]
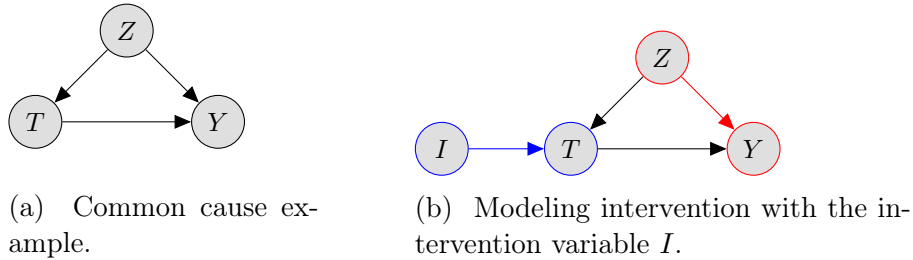

  \begin{center}
    \mbox{\begin{subfigure}{0.25\textwidth}
        \tikz{ %
          \node[obs] (Zm) {$Z$} ; %
          \node[obs, below left=of Zm] (Xm) {$T$} ; %
          \node[obs, below right=of Zm] (Ym) {$Y$} ; %
          \edge  {Zm}  {Xm} ; %
          \edge {Zm} {Ym} ; %
          \edge {Xm} {Ym} ; %
        } 
        \begin{center}
          \caption{\label{fig:commonCauseExample} Common cause example.}
        \end{center}
      \end{subfigure}}
    \hspace{1cm}
    \mbox{\begin{subfigure}{0.4\textwidth}
        \tikz{ %
          \node[obs][draw=red] (Zm) {$Z$} ; %
          \node[obs, below left=of Zm][draw=blue] (Xm) {$T$} ; %
          \node[obs, left=of Xm][draw=blue] (Im) {$I$} ; %
          \node[obs, below right=of Zm][draw=red] (Ym) {$Y$} ; %
          \edge {Zm} {Xm} ; %
          \edge[draw=red] {Zm} {Ym} ; %
          \edge {Xm} {Ym} ; %
          \edge[draw=blue] {Im} {Xm} ; %
        }
        \caption{      \label{fig:commonCauseWithIntervention} Modeling intervention with the intervention variable
          $I$.}
      \end{subfigure}}
  \end{center}
  \caption{Common cause\label{fig:simple} (Example~\ref{ex:common-cause}).}
\end{figure}

\begin{example}[Common cause]
  \label{ex:common-cause}
  First, to better understand how DAGs, and more generally SCMs, can be modeled
  with information fields, we provide a  detailed instance for a set of random
  variables that can be represented by the DAG in Figure~\ref{fig:simple}. 
  \textbf{Such an effort is not required in practice}, because the measurability
  properties are
  fully specified by the DAG for such a simple instance. 
  Let $\AGENT =\{Z,T,Y\}$.
  To simplify the exposition, we suppose that the values of each of the three random variables represented on the DAG
  belong to  $\{0,1\}$.
  Then, 
  \( \CONTROL_Z=\CONTROL_T=\CONTROL_Y=\{0,1\} \),
  each equipped with the complete field
  \( \tribu{\Control}_{Z}=\tribu{\Control}_{T}=\tribu{\Control}_{Y}
  = \ba{\emptyset, \na{0}, \na{1}, \na{0,1} } \).
  We take \( \Omega=\{0,1\}^3 \) as Nature set, 
  equipped with the complete field
  \( \tribu{\NatureField}= 2^{\Omega} \) made of all subsets of~$\Omega$. 
  We write \( \Omega=\Omega_Z\times\Omega_T\times\Omega_Y \),
  where \( \Omega_Z=\Omega_T=\Omega_Y=\{0,1\} \),
  and \( \tribu{\NatureField}= \tribu{\NatureField}_Z\otimes
  \tribu{\NatureField}_T\otimes \tribu{\NatureField}_Y \),
  where \( \tribu{\NatureField}_Z=\tribu{\NatureField}_T=\tribu{\NatureField}_Y
  = \ba{\emptyset, \na{0}, \na{1}, \na{0,1} } \).
  To represent, for instance, the arrows pointing to $Y$ in the DAG in Figure~\ref{fig:commonCauseExample}
  (as well as implicit assumptions about information on Nature),
  we require that the information field \( \tribu{\Information}_{Y} \) satisfies
  $
  \tribu{\Information}_{Y} 
  \subset 
  \oproduit{ \{\emptyset,\Omega_Z\}\otimes
  \{\emptyset,\Omega_T\}\otimes
  \tribu{\NatureField}_Y }
{ \tribu{\Control}_{Z} \otimes 
  \tribu{\Control}_{T} \otimes 
  \{\emptyset,\CONTROL_{Y}\} }
  $.
  This relation expresses that
  the information of $Y$ depends at most on
  its own ``source of uncertainty'' (the field $\tribu{\NatureField}_Y$)
  and on the decisions of both $Z$ and $T$ 
  (the field $\tribu{\Control}_{Z}\otimes \tribu{\Control}_{T} $).
  Again, the effort of describing explicitly the information field is
  not required in the case of DAGs, because the mapping from DAGs to IDMs
  is trivial. On the other hand the IDM allows to express more
  sophisticated hypotheses.
\end{example}

\subsection{Conditional precedence (informal definition)}
\label{sec:conditional}

We now exploit the flexibility of the concept of information field to extend the
definition of precedence.  For any subset $\Bgent\subset\AGENT$, let
\( \tribu{\History}_\Bgent = \oproduit{ \tribu{\NatureField} }{ \bigotimes
\limits_{\bgent \in \Bgent} \tribu{\Control}_\bgent } \subset \tribu{\History} \).
In our extended definition of SCM --- 
the Information Dependency Model of Definition~\ref{de:Information_Dependency_Model} --- we do not
specify a precedence relation: the primitives are the information fields, and
the notion of precedence is deduced from those fields.  For instance, the
traditional \emph{precedence relation} 
on $\AGENT$ is now written as
\begin{equation}
  \Precedence\agent =
  \bigcap_{\Bgent \subset \AGENT; \tribu{\Information}_{\agent} \subset
    \tribu{\History}_{\Bgent}}\!\!\!\!\!\!\!\!\Bgent
  \quad\mtext{ or, equivalently, }
  \tribu{\Information}_{\agent} \subset \tribu{\History}_{\Bgent}
  \iff \Precedence\agent\subset\Bgent 
  \eqfinp
  \label{eq:precedence_relation}
\end{equation}
For an SCM satisfying Equation~\eqref{eq:SCM}, using the mapping
to an IDM described in Remark~\ref{remark:scmtoidm},
one can check that parental and precedence relations are related by
$P(a) = \Precedence \agent$ when $P(a)$ is
  the smallest set such that Equation~\eqref{eq:SCM} is satisfied:
the relation $\Control_\agent\np{\omega} =
\wstrategy_\agent\bp{\omega_\agent,\Control_{P(a)}\np{\omega}} $
implies that $\tribu{\Information}_{\agent} \subset
\tribu{\History}_{P(a)}$; moreover, the minimality means that $P(a)$ is the smallest
subset of $\AGENT$ satisfying such constraint.
So if the SCM can be represented by a  DAG, $\bgent\in\Precedence \agent$ means that there is an
arrow from~$\bgent$ to~$\agent$ in this DAG.

Here is how the notion of information field allows to extend the definition of precedence
to conditional precedence.

\begin{definition}[Conditional Precedence]
  \label{de:conditional_precedence_relation}
  For any subset~$\HistorySubset\subset \HISTORY$ of configurations, 
  and any subset~$\AgentSubsetW\subset\AGENT$, we set
  \begin{equation*}
    \PrecedenceWH \agent = 
    \bigcap_{\Bgent\subset\AGENT; \tribu{\Information}_{\agent}\cap \HistorySubset
      \subset \tribu{\History}_{\Bgent\cup \AgentSubsetW}\cap \History}\!\!\!\!\!\!\!\!\!\!\!\!\!\!\Bgent
    \eqsepv \forall \agent\in\AGENT
    \eqfinp
  \end{equation*}
  Then, we call~$\PrecedenceWH$ the \emph{precedence conditioned on $(\AgentSubsetW,\History)$}
  binary relation on~$\AGENT$ defined by
  \( \bgent\PrecedenceWH\agent \iff \bgent\in\PrecedenceWH\agent \).
\end{definition}
Informally, $\History$ and $\AgentSubsetW$ are elements over which we ``condition''
(recall that \( \tribu{\Information}_{\agent}\cap \HistorySubset \) and
\( \tribu{\History}_{\Bgent\cup \AgentSubsetW}\cap \History \) are trace fields over~$\HistorySubset$).
When \( \AgentSubsetW=\emptyset \) and \( \History=\HISTORY \),
we get that \( \PrecedenceWH=\Precedence \). 

\begin{example}[Recursive Information Dependency Model]
  Informally an Information Dependency Model is  recursive when it
  corresponds to a DAG, i.e.
  $\Precedence=\Precedence_{\emptyset,\HISTORY}$ induces a
  partial order. 
\end{example}

\begin{remark}[Solvability]
  \label{rq:solvability}
  When the Information Dependency Model described
  by~\eqref{eq:decision_rule} has been constructed from a DAG (see
  Remark~\ref{remark:scmtoidm}), there is no question of
  well-posedness. Indeed, one can
  simulate a sample of random variables by first generating the variables that
  do not have parents, and then following the graph along their
  children.
 In such situation, we can equivalently say  that the IDM allows  for a
 sequential resolution, is 
recursive, or admits a fixed causal ordering~\cite{Witsenhausen:1971a}.

However, there exist 
   IDMs that do not have a fixed causal ordering. For such
   \emph{nonsequential} (or equivalently, \emph{nonrecursive}) IDMs,
 we  require a weaker property than sequentiality to ensure
 well-posedness: \emph{solvability}. 
   We discuss in more details the question of
  \textbf{solvability} in~\S\ref{Cycles_and_well-posedness}.  We need in
  particular to exclude cases such as \emph{self-information} (that is,
  $\agent \in P(\agent)$) and, more generally, cases where the system of equations
  \eqref{eq:SCM} could have several solutions (consider for instance $x= y$ and
  $y=x$) or no solution at all.
\end{remark}

In \citep{tikka2019identifying}, the authors manage to summarize the
three rules of do-calculus into one rule thanks to the notions of context specific
independence and labeled DAGs.
Our definition allows us to reproduce their approach.
\begin{example}[Context Specific Independence]
  \label{example:tikka}
  In order to model spurious edges, 
\citep{tikka2019identifying}  relies
  on so-called \emph{labeled DAGs} that can be turned into a context specific DAG by removing
  the arcs that are deactivated (spurious) in the context of interest.
  In the formalism that we propose, such context is represented by a subset~$\History$ of~$\HISTORY$.
  Indeed, if we denote by $\History\subset\HISTORY$ the context for which an arc~\(
  \np{\agent,\bgent} \) is deactivated (in the language of
  \citep{tikka2019identifying}),
  we encode this by the two properties
  \( \agent \not\in\Precedence_{\emptyset,\History}\bgent \) and
  \( \agent \in\Precedence_{\emptyset,\HISTORY\setminus\History}\bgent \)
  --- themselves encoded in the structure of $\bgent$'s information
  field~$\tribu{\Information}_{\bgent}$.
\end{example}
For the reader familiar with~\citep{tikka2019identifying}, it is then
easy to guess how we are going to model intervention variables. 
\begin{example}[Intervention variables]
  To  introduce the possibility to intervene on a variable,
  we use a simple procedure.  
  Suppose we are interested in an intervention profile
  $\hat{\wstrategy}_{\AgentSubsetZ}$ for a subset~\( \AgentSubsetZ \subset \AGENT \).
  For this purpose, we consider a new family
  \( \nseqa{\hat{\tribu{\Information}}_{z}}{z\in\AgentSubsetZ} \)
of fields   \( \hat{\tribu{\Information}}_{z} \subset \tribu{\History} \),
  and we suppose that $\hat{\wstrategy}_\AgentSubsetZ$ is
  $\hat{\tribu{\Information}}_{z}$-measurable, for any $z\in \AgentSubsetZ$.
  Then, we enrich the model as follows:
  (i) we introduce a new \emph{intervention variable}~$I$ (see Figure~\ref{fig:commonCauseWithIntervention}),
  equipped with $\Omega_I =\{0,1\}$ and $\CONTROL_I = \{0,1\}$, 
  and which only has access to its private information in~$\Omega_I$;
  (ii) we straightforwardly adapt the information fields 
  for $\AGENT\setminus\AgentSubsetZ$ and the probability~$\PROBA$;  
  (iii) we replace the information field $\tribu{\Information}_{z}$ by
  $\bp{ \{0\}\otimes\tribu{\Information}_{z} } \cup
  \bp{ \{1\}\otimes\hat{\tribu{\Information}}_{z} }$, for $z\in \AgentSubsetZ$.

  More formally, we introduce  the new model
  $\bp{
    \tilde {\AGENT},
    \np{\tilde{\Omega}, \tilde{\tribu{\NatureField}}}, 
    \nseqa{\tilde{\CONTROL}_{\agent}, \tilde{\tribu{\Control}}_{\agent}}{\agent \in \tilde{\AGENT}}, 
    \nseqa{\tilde{\tribu{\Information}}_{\agent}}{\agent \in \tilde{\AGENT}} 
  }$, where
  $\tilde{\AGENT}= \AGENT \cup \{I\}$, 
  \(\tilde{ \Omega} = \Omega\times \{0,1\}\), 
  $\tilde{\CONTROL}_I = \{0,1\}$,
  $\tilde{\CONTROL}_{\agent}=\CONTROL_{\agent} $ for any \( \agent \in \AGENT \),
  and  
  \begin{subequations}
    \begin{align}
      \tilde{\tribu{\Information}}_{\agent} 
      &=
        \tribu{\Information}_{\agent}
        \otimes \{\emptyset,\CONTROL_I\}
        \eqsepv   \forall \agent \in \AGENT \setminus \AgentSubsetZ
        \eqfinv
      \\
      \tilde{\tribu{\Information}}_{\AgentSubsetZ} 
      &= 
        \hat{\tribu{\Information}}_{\AgentSubsetZ}
        \otimes \tribu{\Control}_I
        \eqsepv
        \forall z\in\AgentSubsetZ
        \eqfinv
      \\
      \tribu{\Information}_{I} 
      &=   
      \oproduit{ \bigotimes_{\agent\in\AGENT} \{\emptyset, \Omega_\agent\}
        \otimes \tribu{\NatureField}_I 
        }{
        \bigotimes_{\agent\in\tilde{\AGENT}} \{\emptyset, \CONTROL_\agent\} }
        \eqfinp 
    \end{align}
  \end{subequations}
  We also extend the probability~$\PROBA$ as a product probability~$\tilde{\PROBA}=\PROBA
  \otimes \mu$ on~$\tilde{ \Omega}$, where $\mu$ is a full support probability on~$\{0,1\}$. 
\end{example}

\section{Topological separation, independence and do-calculus}
\label{sec-Main-results}

In~\S\ref{sec:topological-separation}, we introduce the new notion of
topological separation.
In~\S\ref{sec:independence_and_do-calculus_with_information_fields},
we prove that topological separation implies 
independence, which allows us to derive  a unique do-calculus rule. 
We use the notation \( \ic{r,s}=\na{r,r+1,\ldots,s-1,s} \) for any two
integers~$r,s$ such that $r \leq s$. 

\subsection{Definition of topological separation}
\label{sec:topological-separation}

We now introduce the new notion of topological separation.
We refer the reader to the companion
paper~\cite{De-Lara-Chancelier-Heymann-2021}
for additional material on the subject. 

For any subsets \( \Bgent \subset \AGENT\) and \( \Bgent_j \subset \AGENT\),
$j\in\ic{1,n}$, we write \( \Bgent_1 \sqcup \cdots \sqcup \Bgent_n = \Bgent \)
when we have both \( \Bgent_j \cap \Bgent_k = \emptyset \) for all $j\neq k$ and
\( \Bgent_1 \cup \cdots \cup \Bgent_n = \Bgent \).  We will also say that
\( \sequence{\Bgent_j}{j\in\ic{1,n}} \) is a \emph{splitting} of~\( \Bgent \)
(we do not use the vocable of partition because it is not required that the
subsets~\( \Bgent_j \) be nonempty).

\begin{definition}[Topological Separation]
  \label{de:topologically_separated}
  Let $\HistorySubset\subset \HISTORY$ and  $\Bgent,
  \Cgent,\AgentSubsetW\subset\AGENT$.
  We denote by $\TopologicalClosure{\Bgent}$ the smallest subset of~$\AGENT$
  that contains $\Bgent$ and its own predecessors under~$\PrecedenceWH$
  (that is, \( \TopologicalClosure{\Bgent}=
  \Bgent \cup \PrecedenceWH\Bgent \cup \PrecedenceWH^2\Bgent \cup \cdots \)). 
  
  We say that  $\Bgent$ and $\Cgent$ are 
  (conditionally) \emph{topologically separated} 
  (\wrt\footnote{with respect to} $\np{\AgentSubsetW,\HistorySubset}$), 
  denoted by $\Bgent \ConditionalTopologicalSeparation \Cgent \mid \np{\AgentSubsetW,\HistorySubset}$,
  if there exists 
  \( \AgentSubsetW_\Bgent , \AgentSubsetW_\Cgent \subset \AgentSubsetW \) 
  such that 
  \begin{equation}
    \AgentSubsetW_\Bgent \sqcup \AgentSubsetW_\Cgent = \AgentSubsetW 
    \quad\text{and}\quad
    \TopologicalClosure{ \Bgent \cup \AgentSubsetW_\Bgent }
    \cap 
    \TopologicalClosure{ \Cgent \cup \AgentSubsetW_\Cgent }
    = \emptyset 
    \eqfinp
    \label{eq:topologically_separated_TopologicalClosure}      
  \end{equation}
\end{definition}
Further discussion on this definition is provided in
Sect.~\ref{subsec:conditonal-parentality}.  As proved in
Proposition~\ref{pr:conditional_topology} (see also \citep{Witsenhausen:1975}),
$\TopologicalClosure{\Bgent}$ is the \textbf{topological closure} of $\Bgent$
under a topology induced by the relation~$\PrecedenceWH$.  As stated in the
introduction, we prove in~\citep{De-Lara-Chancelier-Heymann-2021} that the
topological separation is equivalent to the d-separation on DAGs.

Observe that the condition for topological separation is on the \emph{existence}
of a splitting of the set of variables $\AgentSubsetW$ over which we want to
condition.  On a DAG, when $\HistorySubset= \HISTORY$, we have topological
separation of $\Bgent$ and $\Cgent$ with respect to $\AgentSubsetW$ when there
is a splitting $(\AgentSubsetW_\Bgent,\AgentSubsetW_\Cgent)$ of~$\AgentSubsetW$
such that the sets of ancestors of $\Bgent\cup \AgentSubsetW_\Bgent$ and
$\Cgent\cup \AgentSubsetW_\Cgent$ ---  defined as the union of the
iterates of the set-valued mapping $\Bgent\subset\AGENT\to
P(\Bgent) \setminus \AgentSubsetW $ on those sets --- are
disjoint.

%

We prove in~\citep{De-Lara-Chancelier-Heymann-2021}
that d-separation and topological separation are
equivalent.  We think this alternative definition of d-separation is very handy even for DAGs.
Indeed, (i) the splitting $(\AgentSubsetW_\Bgent,\AgentSubsetW_\Cgent)$
is given by an explicit formula
(see~\citep{De-Lara-Chancelier-Heymann-2021}),
(ii) once the splitting is given --- which is the main difficulty --- it is usually  much quicker to check
that the ancestors sets are disjoints than checking that all the paths
between~$\Bgent$ and~$\Cgent$ are blocked by~$\AgentSubsetW$,
as illustrated in the following Example~\ref{ex:easy}, illustrated by Figure~\ref{fig:kuh}.
\begin{figure}[h]
  \begin{center}
    \mbox{
      \includegraphics[width=\textwidth]{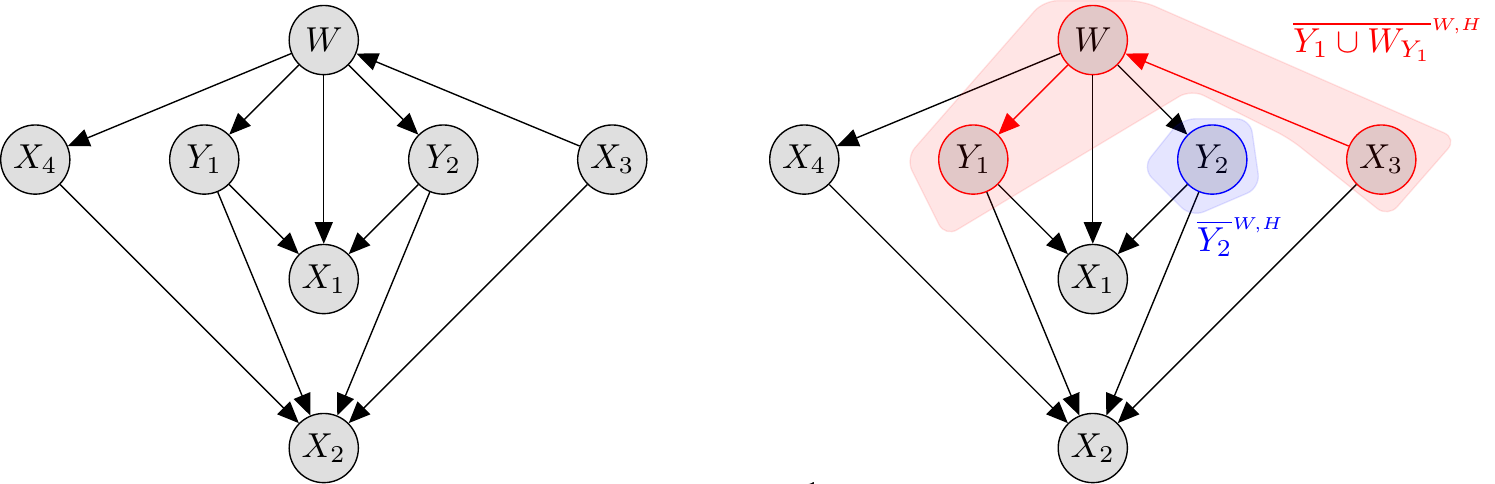}
    }
    \caption{\label{fig:kuh}Topological separation \wrt
      $\np{\AgentSubsetW,\HISTORY}$ is easy to check given
      the splitting $W = W_{Y_1}\sqcup W_{Y_2}$, with $W_{Y_1} = W$ and $W_{Y_2}=\emptyset$
      as $\TopologicalClosure{Y_1\cup W_{Y_1}}=
      \TopologicalClosure{Y_1\cup W}=Y_1\cup W\cup X_3$
      --- with, in red, the edges followed to build the closure,
      and the three vertices~$Y_1,W,X_3$ where~$X_3$ is the 
      only new vertex in $\TopologicalClosure{Y_1\cup W_{Y_1}}
      \setminus \np{Y_1\cup W}$ --- and
      $\TopologicalClosure{Y_2\cup W_{Y_2}}=\TopologicalClosure{Y_2}= Y_2$
      --- with, in blue, the only vertex~$Y_2$ because only~$W$ has an arrow
      pointing to~$Y_2$, hence has to be excluded since it is in~$W$
      --- do not intersect. 
    }
  \end{center}
\end{figure}

\begin{example}[Topological separation is easy to check: recursive system]
  \label{ex:easy}
  The DAG in Figure~\ref{fig:kuh} (left) illustrates why and how the notion of
  topological separation is practical.
  If one wants to check that $Y_1$ and $Y_2$ are d-separated by $W$, then one 
  needs to check that \textbf{every path} that goes from $Y_1$ to $Y_2$ is
  blocked by $W$ (by simply applying the definition), like in
  \begin{itemize}
  \item $Y_1\leftarrow W\to Y_2$: blocked common cause
  \item $Y_1\to  X_1\leftarrow  Y_2$: collider
  \item $Y_1\to  X_2\leftarrow  Y_2$: collider
  \item $Y_1\to  X_1\leftarrow  W \to Y_2$: blocked common cause
  \item $Y_1\to X_2\leftarrow X_4\leftarrow W \to Y_2$: collider
  \item $\dots$
  \end{itemize}
  Of course, one could simplify this long enumerating process by observing,
  for instance, that any path going through~$X_2$ will be blocking, but this requires additional steps.  

  By contrast, the topological separation can be checked
  visually on the right hand side of Figure~\ref{fig:kuh} by setting $\AgentSubsetW_{Y_1}= \AgentSubsetW$,
  $\AgentSubsetW_{Y_2}= \emptyset$ and then checking that the
  topological closures of $Y_1\cup W$ (which is $Y_1\cup W\cup X_3$) and $Y_2$
  (which is $Y_2$ itself)
do not intersect.
\end{example}

\begin{figure}
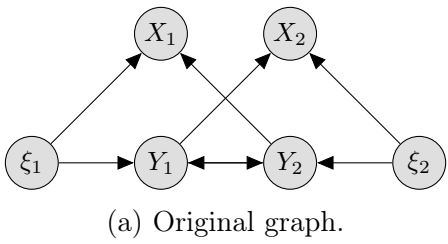
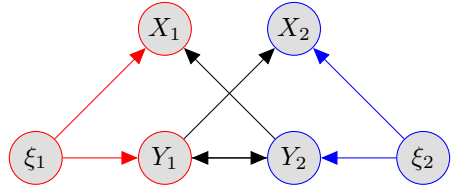

  \begin{center}
    \mbox{\begin{subfigure}{.45\textwidth}
        \centering
        \tikz{ %
          \node[obs] (X1) {$X_1$} ; %
          \node[obs,  right=of X1] (X2) {$X_2$} ; %
          \node[obs, below=of X1] (Y1) {$Y_1$} ; %
          \node[obs, below=of X2] (Y2) {$Y_2$} ; %
          \node[obs, left=of Y1] (Xi1) {$\xi_1$} ; %
          \node[obs, right=of Y2] (Xi2) {$\xi_2$} ; %
          \edge  {Xi1}  {X1} ; %
          \edge  {Xi2}  {X2} ; %
          \edge  {Xi1}  {Y1} ; %
          \edge  {Xi2}  {Y2} ; %
          \edge  {Y2}  {X1} ; %
          \edge  {Y1}  {X2} ; %
          \edge  {Y2}  {Y1} ; %
          \edge  {Y1}  {Y2} ; %
        }
        \caption{Original graph. \\ \quad}
        \label{fig:sub-first}
      \end{subfigure}}\hspace{0.05\textwidth}%
    \mbox{\begin{subfigure}{.45\textwidth}
        \centering
        \tikz{ %
          \node[obs][draw=red] (X1) {$X_1$} ; %
          \node[obs,  right=of X1][draw=blue] (X2) {$X_2$} ; %
          \node[obs, below=of X1][draw=red]  (Y1) {$Y_1$} ; %
          \node[obs, below=of X2][draw=blue] (Y2) {$Y_2$} ; %
          \node[obs, left=of Y1][draw=red] ( (Xi1) {$\xi_1$} ; %
          \node[obs, right=of Y2][draw=blue] ( (Xi2) {$\xi_2$} ; %
          \edge[draw=red]  {Xi1}  {X1} ; %
          \edge[draw=blue]  {Xi2}  {X2} ; %
          \edge [draw=red] {Xi1}  {Y1} ; %
          \edge[draw=blue]  {Xi2}  {Y2} ; %
          \edge  {Y2}  {X1} ; %
          \edge  {Y1}  {X2} ; %
          \edge  {Y2}  {Y1} ; %
          \edge  {Y1}  {Y2} ; %
        }
        \caption{Let $\AgentSubsetW_{X_i} = Y_i$, for $i=1,2$.  The closure of
          $X_1\cup Y_1$ (resp. $X_2\cup Y_2$), with the edges followed to build the
        closure, is in red (resp. blue).}
        \label{fig:sub-fourth}
      \end{subfigure}}
    \caption{\label{fig:jpcbh}Topological separation is easy to check: nonrecursive system.
    }
  \end{center}
\end{figure}

\begin{example}[Topological separation is easy to check: nonrecursive system]
  \label{ex:non-rec-easy}
  We display in Figure~\ref{fig:jpcbh} a nonrecursive system for which
  we check that $X_1$ and $X_2$ are topologicaly separated
    \wrt $\bp{\np{Y_1,Y_2},\HISTORY}$.
  This is --- in our opinion --- simpler to check than 
  $\sigma$-separation~\cite{pmlr-v115-forre20a} because there are less intermediate steps. 
\end{example}

\subsection{Independence and do-calculus with information fields}
\label{sec:independence_and_do-calculus_with_information_fields}

In what follows,
we consider random variables
that take values in countable sets, and that $\Omega$ is countable as
well.
We pinpoint that  working on continuous sets introduces technical measurability questions in the
proof of Lemma~\ref{th:decomposition} which are likely to be
irremediable, as discussed in~\S\ref{subseq:discrete-continuous}.
We can now state our version of Pearl's three rules of do-calculus.
The statement looks like a simple sufficient condition for conditional
independence thanks to the fact that we encode the  intervention variables
in the information fields.

\begin{theorem}[Do-calculus]
  \label{th:do-calculus3rule-simple}
  Supposing that all random variables 
  have countable codomain, and that $\Omega$ is countable as
well, we have the implication
  \begin{equation}
    \AgentSubsetY \Tsep \AgentSubsetZ \mid \np{\AgentSubsetW,\HistorySubset} \implies
    \PROBA
    \bsetp{ U_\AgentSubsetY=\cdot}{U_\AgentSubsetW,U_{\TopologicalClosure{\AgentSubsetZ}},\HistorySubset}
    = \PROBA
    \bsetp{ U_\AgentSubsetY=\cdot}{U_\AgentSubsetW,\HistorySubset}    
    \eqfinp 
  \end{equation}
\end{theorem}
A more formal statement of this Theorem~\ref{th:do-calculus3rule-simple}, as well as the proof and a discussion,
will be provided in~\S\ref{Topological_separation_implies_the_do-calculus}. 
We stress the conciseness of Theorem~\ref{th:do-calculus3rule-simple}
-- a unique rule, no ``do'' operator --  
permitted by the Information Dependency Model formulation.

\begin{figure}[h]
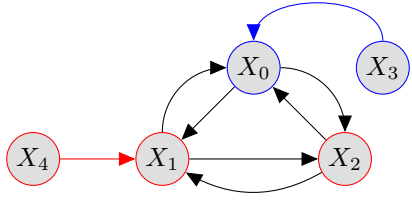

  \begin{center}
    \tikz{ %
      \node[obs][draw=blue] (X0) {$X_0$} ; %
      \node[obs, below left=of X0][draw=red] (X1) {$X_1$} ; %
      \node[obs, below right=of X0][draw=red] (X2) {$X_2$} ; %
      \node[obs,  right=of X0][draw=blue] (X3) {$X_3$} ; %
      \node[obs,  left=of X1][draw=red] (X4) {$X_4$} ; %
      \edge  {X2}  {X0} ; %
      \edge  {X0}  {X1} ; %
      \edge[draw=red]   {X4}  {X1} ; %
      \draw [->] (X1) to [out=90,in=180] (X0);
      \draw [->] (X1) to [out=0,in=180] (X2);
      \draw [<-] (X1) to [out=330,in=210] (X2);
      \draw [->] (X0) to [out=0,in=90] (X2);
      \draw [->,blue] (X3) to [out=90,in=90] (X0);
    } 
    \caption{\label{fig:CE2}$X_3$ and $X_4$ are independent
      conditioned on $(X_0,X_1,X_2)$ but not independent if we only
      condition on $(X_0,X_1)$.
The visual proof of topological separation is obtained by considering the
        splitting $W_{X_4}=\{X_1,X_2\}$ and $W_{X_3}=\{X_0\}$ and observing that the topological closure of
      $X_3 \cup W_{X_3}$ in blue does not intersect the topological closure of $X_4 \cup W_{X_4}$ in red}.
  \end{center}
\end{figure}

\begin{example}
  \label{ex:w}
  The following, well posed (solvable as given below in Definition~\ref{de:solvability}), example is inspired by the work of Witsenhausen~\cite{Witsenhausen:1971a}
  on causality. It is depicted
  in Figure~\ref{fig:CE2} and corresponds to the following nonrecursive binary SCM
  ($N_1$,\ldots, $N_4$ are independent, binary noise variables, $\oplus$ is the $\mathrm{Xor}$
  operator):
  \begin{align*}
    X_0 &= (X_1. (\neg X_2))\oplus (N_0 \oplus X_3)\quad\mbox{and} \quad
          X_1 = (X_2 .(\neg X_0))\oplus (N_1 \oplus X_4)\\
    X_2 &= (X_0. (\neg X_1))\oplus N_2, \quad 
          X_3 = N_3 \quad \mbox{and} \quad 
          X_4 = N_4 \eqfinp
  \end{align*}
  The random variables  $X_3$ and $X_4$ are
  topologically separated by $(X_0,X_1,X_2)$ --- note that $X_2$ is needed ---
  \textbf{hence $X_3$ and $X_4$ are independent conditioned on $(X_0,X_1,X_2)$}
  but not independent if we only condition on $(X_0,X_1)$.

  Observe that the intuition that we could equivalently replace $X_0$, $X_1$ and
  $X_2$ by a unique variable $\AgentSubsetW$ is misleading: with such a change, we
  would get a collider $X_4 \to \AgentSubsetW \leftarrow X_3$ over which we are
  conditioning, which would make $X_4$ and $X_3$ non-blocked with
  respect to $\AgentSubsetW$.

  Let us try to apply the elegant recent result
  \cite[Theorem~5.2]{pmlr-v115-forre20a}
  on conditional independence in the presence of cycles.
  We first observe that 
  the Directed Mixed Graph (DMG) induced by the Input/output Structural Causal
  Model (ioSCM) associated with our example (see
  \cite[Definitions~2.3 and 5.1]{pmlr-v115-forre20a}) looks like the graph of
  Figure~\ref{fig:CE2}.
  Second, we observe that $X_0$, $X_1$ and $X_2$ belong to the same
  strongly connected component~$S$  (see~\cite{pmlr-v115-forre20a}), in the
  sense that they are all ancestors and descendants of each other.
  Third, let us consider the walk
  $X_4\to X_1\leftarrow X_0 \leftarrow X_3$.
  According to \cite[Definition 4.2]{pmlr-v115-forre20a},
  the walk $X_4\to X_1\leftarrow X_0  \leftarrow X_3$ is 
  $\{X_0,X_1,X_2\}$-\emph{$\sigma$-open} because
  \begin{itemize}
  \item $
    X_4\to X_1\leftarrow X_0$ satisfies the  collider definition
    in \cite[Definition~4.2, (a)]{pmlr-v115-forre20a}, as
    $X_1\in \{X_0,X_1,X_2\}$,
  \item
    $ X_1\leftarrow X_0 \leftarrow X_3$ satisfies the left chain
    condition because $X_0 \in \{X_0,X_1,X_2\}\cap S $, where $S$ is
    the strongly connected component of $X_1$.
  \end{itemize}
  Hence it seems that \cite[Theorem~5.2]{pmlr-v115-forre20a} cannot be used to state that $X_3$
  and $X_4$ are independent conditioned on $(X_0,X_1,X_2)$.

  To finish, we illustrate the claim --- that is, the independence
  of~$X_3$ and $X_4$ when conditioned on~$X_0$, $X_1$ and $X_2$ ---
  with a numerical exact computation, taking the
  $N_i$ as binomial variables  of parameter~$0.1$.
  We solve the cycle by enumerating the 8 possible combinations of values
  for~$X_0$, $X_1$ and $X_2$ and selecting the only admissible one.
  The results are shown in Table~\ref{table:table}.

  This  example illustrates the novelty of the IDM approach.
\end{example}
\newcommand{\verteq}{\rotatebox{90}{$\,=$}}
\newcommand{\equalto}[1]{\underset{\scriptstyle{\mkern4mu\verteq}}{#1}}
\newcommand{\toequal}[1]{\underset{\displaystyle #1}{{\mkern4mu\verteq}}}

\begin{table}[!htb]
  \caption{Numerical results for Example~\ref{ex:w} \label{table:table}}
  \begin{subtable}[t]{.45\textwidth}
    \centering
    \caption{We check numerically that 
      $X_3$ and $X_4$ are independent when conditioned on
      $X_0$, $X_1$ and $X_2$ by computing $\PROBA( X_4 = 1 | X_0,X_1,X_2,X_3 )$;
      indeed, the last two columns are identical.}
    \begin{tabular}{@{}ccc||l|l@{}}
      \toprule
      $\equalto{X_0}$ & $\equalto{X_1}$ & $\equalto{X_2}$ & $X_3={0}$  & $X_3=1$  \\ \midrule
      {0} & {0} & {0} & 0.012        & 0.012        \\
      {0} & {0} & {1} & 0.5          & 0.5          \\
      {0} & {1} & {0} & 0.5          & 0.5          \\
      {0} & {1} & {1} & 0.012        & 0.012        \\
      {1} & {0} & {0} & 0.012        & 0.012        \\
      {1} & {0} & {1} & 0.012        & 0.012        \\
      {1} & {1} & {0} & 0.5          & 0.5          \\
      {1} & {1} & {1} & 0.5          & 0.5          \\ \bottomrule
    \end{tabular}
  \end{subtable}%
  \hspace{0.05\textwidth}\begin{subtable}[t]{.45\textwidth}
    \centering
    \caption{We check numerically that 
      $X_3$ and $X_4$ are not independent when conditioned on
      $X_0$ and $X_1$ by computing $\PROBA( X_4 = 1 | X_0,X_1,X_3 )$;
    indeed, the last two columns are different (see the two underlined numbers).
      \\
    }
    \begin{tabular}{cc||l|l}
      \toprule
      $\equalto{X_0}$ & $\equalto{X_1}$ & $X_3=0$  & $X_3=1$  \\ \midrule
      {0} & {0} & 0.023        & 0.023        \\
      {0} & {1} & \underline{0.1}          & \underline{0.474}        \\
      {1} & {0} & 0.012        & 0.012        \\
      {1} & {1} & 0.5          & 0.5          \\ \hline
    \end{tabular}
  \end{subtable} 
\end{table}

\section{Presentation of Witsenhausen's product model and solvability}
\label{Presentation_of_Witsenhausen_product_model_and_solvability}


This Sect.~\ref{Presentation_of_Witsenhausen_product_model_and_solvability}
is devoted to a presentation of  
the mathematical formalism and technical machinery we rely on, which we borrow from Witsenhausen's
work~\citep{Witsenhausen:1971a,Witsenhausen:1975}.
We start in~\S\ref{The_Witsenhausen_product_model} with Witsenhausen's product
model, followed by the notions of solvability and solution map
in~\S\ref{Solvability_and_solution_map}.
It is notable that our work brings   together ideas from causal statistics with ideas from   
decentralized control theory that also attempted to provide a definition of
causality a few decades ago; this is the object of~\S\ref{Causality}.
Thus equiped, we discuss cycles and the meaning of ``well-posedness''
in~\S\ref{Cycles_and_well-posedness}.

\subsection{Witsenhausen's product model}
\label{The_Witsenhausen_product_model}

Because Witsenhausen introduced his model to 
the control community some five decades ago
\citep{Witsenhausen:1971a,Witsenhausen:1975}, we expect that
most readers will not be familiar with it.
We provide tentative correspondences between Pearl's DAG and Witsenhausen's intrinsic model 
in Table~\ref{tab:Pearl_Witsenhausen}.

\begin{table}[h]
  \begin{tabular}{@{}lll@{}}
    \toprule
    & \emph{Pearl} & \emph{Witsenhausen}   
    \\ \midrule
    Structure           & DAG            & Nature and agents decision sets, with their respective fields 
    \\
    Parent relation     & $\rightarrow$  & precedence relation
    \\
    & node           & agent                                                                                                
    \\
    & edge   & agents related by the precedence relation
    \\
    Dependence & SCM            & agents information fields
    \\
    & functional relation    & policy profiles measurable \wrt\ information fields
    \\
    Resolution          & induction    & solution map     
    \\
    &   random variable         & policy profile composed with solution map   
    \\
    Intervention        & do operator    & change of information fields 
    \\
    Causal ordering     & fixed          & existence depends on agents information fields                                            
    \\ \bottomrule
  \end{tabular}
  \caption{Correspondences between Pearl's DAG and Witsenhausen's intrinsic model
    \label{tab:Pearl_Witsenhausen}}
\end{table}

\begin{definition}(Adapted from \citep{Witsenhausen:1971a,Witsenhausen:1975})
  \label{de:W-model}
  A \emph{W-model} is a collection
  $(\AGENT$,
  $\nseqa{\CONTROL_{\agent}, \tribu{\Control}_{\agent}}{\agent \in \AGENT},$
  $\np{\Omega, \tribu{\NatureField}},$
  $\nseqa{\tribu{\Information}_{\agent}}{\agent \in \AGENT})$,
  where 
  \begin{itemize}
  \item
    $\AGENT$ is a finite set, whose elements are called \emph{agents};
  \item 
    for any \( \agent \in \AGENT \), $\CONTROL_{\agent}$ is a  set, 
    the \emph{set of decisions} for agent~$\agent$;
    $\tribu{\Control}_{\agent}$ is a field over~$\CONTROL_{\agent}$;
  \item 
    \( \Omega \) is a set made of \emph{states of Nature};
    $\tribu{\NatureField}$ is a  field over~\( \Omega \);
  \item
    for any \( \agent \in \AGENT \), \( \tribu{\Information}_{\agent} \)
    is a subfield of the following product field
    \begin{equation}
      \tribu{\Information}_{\agent} \subset 
      \oproduit{\tribu{\NatureField}}{\bigotimes \limits_{\bgent \in \AGENT} \tribu{\Control}_{\bgent}}
      \eqsepv
      \forall \agent \in \AGENT 
      \label{eq:information_field_agent}
    \end{equation}
    and is called the \emph{information field} of the agent~$\agent$.
  \end{itemize}
\end{definition}
A \emph{countable W-model} is a W-model where all sets
\( \sequence{\CONTROL_{\agent}}{\agent \in \AGENT}\)
and $\Omega$ above are countable,
equipped with the complete $\sigma$-algebras. 

We recall that the \emph{configuration space}~$\HISTORY$ and
the \emph{configuration field}~$\tribu{\History}$ are 
\begin{equation}
  \label{eq:HISTORY}
  \HISTORY = \produit{\Omega}{ \prod \limits_{\agent \in \AGENT} \CONTROL_{\agent}}
  \eqsepv
  \tribu{\History} =\oproduit{\tribu{\NatureField}}{\bigotimes \limits_{\agent \in \AGENT}
  \tribu{\Control}_{\agent}} 
  \eqfinp 
\end{equation}
A  \emph{policy} 
of agent~$\agent \in \AGENT$ is a mapping
\begin{subequations}
  \begin{equation}
    \wstrategy_{\agent} : \np{\HISTORY,\tribu{\History}} \to
    \np{\CONTROL_{\agent},\tribu{\Control}_{\agent}}
    \mtext{ such that }
    \wstrategy_{\agent}^{-1} (\tribu{\Control}_{\agent})
    \subset \tribu{\Information}_{\agent} 
    \eqfinp    
  \end{equation}
  Hence, any policy~$\wstrategy_{\agent}$ is a mapping from configurations to decisions,
  which satisfies the measurability property~\eqref{eq:decision_rule}, that is,
  any policy of agent~$\agent$ may only depend upon the
  information~$\tribu{\Information}_{\agent}$ available to~$\agent$.
  We denote by $\WSTRATEGY_{\agent}$ the set of all policies of agent $\agent \in \AGENT$.
  A \emph{policy profile} 
  $\wstrategy$ is a collection 
  of policies, one per agent~$\agent \in \AGENT$:
  \begin{equation}
    \wstrategy = \sequence{\wstrategy_{\agent}}{\agent \in \AGENT} 
    \in \prod_{\agent \in \AGENT} \WSTRATEGY_{\agent}
    \eqsepv
    \mtext{where}\,
    \WSTRATEGY_{\agent} = \bset{ 
      \wstrategy_{\agent} : \np{\HISTORY,\tribu{\History}} \to
      \np{\CONTROL_{\agent},\tribu{\Control}_{\agent}} }%
    {  \wstrategy_{\agent}^{-1} (\tribu{\Control}_{\agent})
      \subset \tribu{\Information}_{\agent} } \eqsepv \forall \agent\in\AGENT 
    \eqfinp
    \label{eq:W-strategy_profile}
  \end{equation}
\end{subequations}

\subsection{Solvability and solution map}
\label{Solvability_and_solution_map}

With any policy profile $\wstrategy = \sequence{\wstrategy_{\agent}}{\agent \in \AGENT}
\in \prod_{\agent \in \AGENT} \WSTRATEGY_{\agent} $ we associate the set-valued
mapping
\begin{align}
  \SetValuedReducedSolutionMap_\wstrategy: \Omega
  & \rightrightarrows \prod_{\bgent \in \AGENT} \CONTROL_{\bgent} 
    \label{eq:SetValuedReducedSolutionMap}
    \eqsepv 
  \omega
  \mapsto \Bset{ \sequence{\control_\bgent}{\bgent \in \AGENT} \in 
  \prod_{\bgent \in \AGENT} \CONTROL_{\bgent} }{%
  \control_{\agent} = \wstrategy_{\agent}\bp{\omega,
  \sequence{\control_\bgent}{\bgent \in \AGENT} }
  \eqsepv \forall \agent \in \AGENT}
  \eqfinp
\end{align}

With this definition, we slightly reformulate below
how Witsenhausen introduced solvability.

\begin{definition}(\citep{Witsenhausen:1971a,Witsenhausen:1975})
  \label{de:solvability}
  \begin{subequations}
    The solvable \emph{(measurable)  property} holds true for the W-model of
    Definition~\ref{de:W-model} --- or the W-model is said to be
    \emph{(measurable) solvable} --- when,
    for any policy profile $\wstrategy = \sequence{\wstrategy_{\agent}}{\agent \in \AGENT}
    \in \prod_{\agent \in \AGENT} \WSTRATEGY_{\agent} $, 
    the set-valued mapping~$\SetValuedReducedSolutionMap_{\wstrategy}$
    in~\eqref{eq:SetValuedReducedSolutionMap}
    is a (measurable) mapping whose domain is $\Omega$, that is,
    the cardinality of $\SetValuedReducedSolutionMap_{\wstrategy}\np{\omega}$ 
    is equal to one, for any state of nature $\omega \in \Omega$.
We denote by SM the \emph{solvability measurability} property.

    Thus, under solvability property, for any state of nature $\omega \in \Omega$, 
    there exists one, and only one, decision profile
    $\sequence{\control_\bgent}{\bgent \in \AGENT} \in 
    \prod_{\bgent \in \AGENT} \CONTROL_{\bgent}$
    which is a solution of the \emph{closed-loop equations}
    \begin{equation}
      \control_{\agent} = \wstrategy_{\agent}\bp{\omega, \sequence{\control_\bgent}{\bgent \in \AGENT}}
      \eqsepv
      \forall \agent \in \AGENT
      \eqfinp
      \label{eq:solution_map_IFF}
    \end{equation}
    In case of solvability, we define the \emph{solution map} 
    \begin{equation}
      \SolutionMap_\wstrategy: \Omega \rightarrow \HISTORY
      \eqsepv 
      \SolutionMap_\wstrategy\np{\omega} = \bp{\omega,\ReducedSolutionMap_{\wstrategy}\np{\omega}}
      \label{eq:solution_map}
    \end{equation}
    where \( \ReducedSolutionMap_{\wstrategy}\np{\omega} \)
    is the unique element contained in the image set
    $\SetValuedReducedSolutionMap_{\wstrategy}\np{\omega}$ that is, 
    for all $\sequence{\control_\bgent}{\bgent \in \AGENT} \in 
    \prod_{\bgent \in \AGENT} \CONTROL_{\bgent}$, 
    $\ReducedSolutionMap_\wstrategy\np{\omega} = 
    \sequence{\control_\bgent}{\bgent \in \AGENT} \iff
    \SetValuedReducedSolutionMap_{\wstrategy}\np{\omega} = \na{\sequence{\control_\bgent}{\bgent \in \AGENT}}$.
  \end{subequations}
\end{definition}

Thus, when the solvability property holds true, for each state of Nature, 
there is a single family of decisions compatible with any given policy profile.
This family is the unique solution of the closed-loop
equations~\eqref{eq:solution_map_IFF}.
In some cases, these equations can be solved sequentially
(where the order may depend on the state of Nature, and on the given policy profile).
This is the case when causality holds true.

\subsection{Causality}
\label{Causality}

In his articles \citep{Witsenhausen:1971a,Witsenhausen:1975},
Witsenhausen introduces a notion of causality that relies on 
suitable configuration-orderings.
Here, we introduce our own notations, as they make possible a compact formulation 
of the causality property.

For any finite/countable set~\( \SET \), let \( \cardinality{\SET} \) denote the cardinality
of~\( \SET \). 
Thus, when \( \AGENT \) is finite, 
\( \numAGENT \) denotes the cardinality of the set~\( \AGENT \), that is,
\( \numAGENT \) is the number of agents.
For $k \in \ic{1,\numAGENT}$, let 
\(
\ORDER_{k}=\defset{ \kappa: \ic{1,k} \to \AGENT }%
{ \kappa \mtext{ is an injection} }
\) denote the set of $k$-orderings, that is, injective
mappings from $\ic{1,k}$ to $\AGENT$.
The set \( \ORDER_{\numAGENT} \) is the set of \emph{total orderings} of
agents in $\AGENT$, that is, bijective
mappings from $\ic{1,\numAGENT}$ to $\AGENT$
(in contrast with \emph{partial orderings} in~$\ORDER_{k}$ for $k < \numAGENT$). 
We define the \emph{set of orderings} by
\(
\ORDER= \bigcup_{ k \in \ic{0,\numAGENT}} \ORDER_{k} \),
where \( \ORDER_{0} = \{ \emptyset \} \).
For any $k \in \ic{1,\numAGENT}$, any ordering $\kappa \in \ORDER_{k}$,
and any integer $\ell \le k$, \( \kappa_{\vert \ic{1,\ell}} \)
is the restriction of the ordering~$\kappa$ to the first $\ell$~integers,
and we introduce the mapping
\( \cut_k: \ORDER_{\numAGENT} \rightarrow \ORDER_{k} 
\eqsepv 
\totalordering \mapsto \totalordering_{\vert \ic{1,k} } \)
which yields the restriction of any total ordering of~$\AGENT$ 
to~$\ic{1,k}$.
For any \( k \in \ic{1,\numAGENT} \), and any $k$-ordering~$\kappa \in \ORDER_{k}$,
we define the \emph{range} $\range{\kappa}=\ba{ \kappa(1), \ldots, \kappa(k) }
\subset \AGENT$, 
the \emph{cardinality} $\cardinality{\kappa}=k \in \ic{1,\numAGENT}$,
the \emph{last element} $\LastElement{\kappa}=\kappa(k)\in \AGENT$,
and 
the \emph{restriction} $\FirstElements{\kappa}=
\kappa_{\vert \ic{1,k-1}} \in \ORDER_{k-1}$.
%


%
The following definition of causality
originates from \citep{Witsenhausen:1971a}.
In a causal W-model, there exists a configuration-ordering with the 
property that when an agent is called to play --- as he is the last one in an
ordering --- what he knows cannot depend on decisions made by agents
that are not his predecessors 
(in the range of the ordering under consideration).
For this purpose, we define, 
for any subset $\Bgent \subset \AGENT$ of agents:
\begin{subequations}
  \begin{align}
    \tribu{\Control}_\Bgent 
    &= 
      \bigotimes \limits_{\bgent \in \Bgent} \tribu{\Control}_\bgent
      \otimes
      \bigotimes \limits_{\agent \not\in \Bgent} \{ \emptyset, \CONTROL_{\agent} \}
      \subset
      \bigotimes \limits_{\agent \in \AGENT} \tribu{\Control}_{\agent}
      \eqfinv
      \label{eq:sub_control_field_Bgent}
    \\
    \tribu{\History}_\Bgent 
    &= 
      \oproduit{ \tribu{\NatureField} }{ \tribu{\Control}_\Bgent }
      = \oproduit{ \tribu{\NatureField} }{ \bigotimes \limits_{\bgent \in \Bgent} \tribu{\Control}_\bgent
      \otimes
      \bigotimes \limits_{\agent \not\in \Bgent} \{ \emptyset, \CONTROL_{\agent}
      \} }
      \subset \tribu{\History}
      \eqfinp
      \label{eq:sub_history_field_Bgent}
  \end{align}
\end{subequations}

\begin{definition}(\citep{Witsenhausen:1971a,Witsenhausen:1975})
  \label{de:causality}
  A countable W-model (as in Definition~\ref{de:W-model}) 
  is \emph{causal} if there exists (at least) one
  causal configuration-ordering $\ordering: \HISTORY \to \ORDER_{\numAGENT}$,
  that is, with the property that 
  \begin{equation}
    \label{eq:causality_a}
    \HISTORY_{\kappa}^{\ordering} \cap \History \in 
    \tribu{\History}_{\range{\FirstElements{\kappa}}}
    \eqsepv 
    \forall \History \in \tribu{\Information}_{\LastElement{\kappa}}
    \eqsepv 
    \forall \kappa \in \ORDER 
    \eqfinv
  \end{equation}
  where the subset~$\HISTORY_{\kappa}^{\ordering} \subset \HISTORY$ of
  configurations is defined by 
  \begin{equation}
    \HISTORY_{\emptyset}^{\ordering} = \HISTORY
    \mtext{ and }
    \HISTORY_{\kappa}^{\ordering} =
    \defset{\history \in \HISTORY}{\psi_{\cardinality{\kappa}}\bp{\ordering(\history)} =\kappa}  
    \eqsepv \forall \kappa \in \ORDER
    \eqfinp
    \label{eq:HISTORY_k_kappa}
  \end{equation}
\end{definition}
The set~$\HISTORY_{\kappa}^{\ordering}$ contains all the configurations 
for which the agent~\( \kappa(1) \) is acting first, 
the agent~\( \kappa(2) \) is acting second, \ldots, till 
the last agent~\( \LastElement{\kappa}=\kappa(\cardinality{\kappa}) \) acting at
stage~$\cardinality{\kappa}$.
Hence, otherwise said, causality means that, once we know the first $\cardinality{\kappa}$~agents, 
the information of the agent~$\LastElement{\kappa}$ depends at most
on the decisions of the agents in the range~$\range{\FirstElements{\kappa}}$,
as represented by the subfield (see Equation~\eqref{eq:sub_history_field_Bgent})
\begin{equation}
  \tribu{\History}_{\range{\FirstElements{\kappa}}}
  =
  \tribu{\NatureField} \otimes 
  \bigotimes \limits_{\agent \in \range{\FirstElements{\kappa}}} \tribu{\Control}_{\agent}
  \otimes
  \bigotimes \limits_{\bgent \not\in \range{\FirstElements{\kappa}} } \{
  \emptyset, \CONTROL_\bgent \}
  \subset \tribu{\History}
  \eqfinp
  \label{eq:causality_b}
\end{equation}

In \citep{Witsenhausen:1971a}, Witsenhausen proves that causal W-models are
 solvable measurable (SM). 
\textbf{The reverse is false}:
in~\cite[Theorem~2]{Witsenhausen:1971a}, Witsenhausen exhibits an example of
noncausal W-model that is solvable.

\subsection{Cycles and well-posedness}
\label{Cycles_and_well-posedness}

It is notable that the framework we develop makes it possible to deal with
systems with cycles. Such systems can be usefull for modeling purpose.  For instance, a cyclic SCM can arise as an equilibrium
state of random differential equations~\cite{bongers2018random}.  The foundations
for structural models with cycles, which are laid out
in~\cite{bongers2020foundations} show that 
the existence of cycles raises well-posedness questions.  In
particular, for such   SCMs, does the
system of equations~\eqref{eq:SCM} has a solution? is it unique? 
Witsenhausen introduces a hierachy of systems that we summarize  in
Figure~\ref{convergence1} (in this hierarchy a DAG corresponds to what Witsenhausen called
a sequential system). 
For our purpose, this hierarchy 
could be qualified as ``too strong'', because it requires
the system to admit a unique solution for all policy profiles.
The solvable-measurable (SM) property in Definition~\ref{de:solvability}
however can be relaxed as soon as we know more about the assignement mappings
under scrutiny.  For instance,
in
\cite{Andersland-Teneketzis:1992} a property --- deadlock-freeness --- weaker
than causality but stronger than SM is identified. Then in 
\cite{Andersland-Teneketzis:1994}  a relaxation of deadlock-freeness
is discussed: instead of imposing a constraint that need to be
satisfied for every admissible policy, the authors propose to put the
strain only on a policy of interest.  We leave for further work a
comparison between how well-posedness is handled
in~\cite{bongers2020foundations} and in the present work.
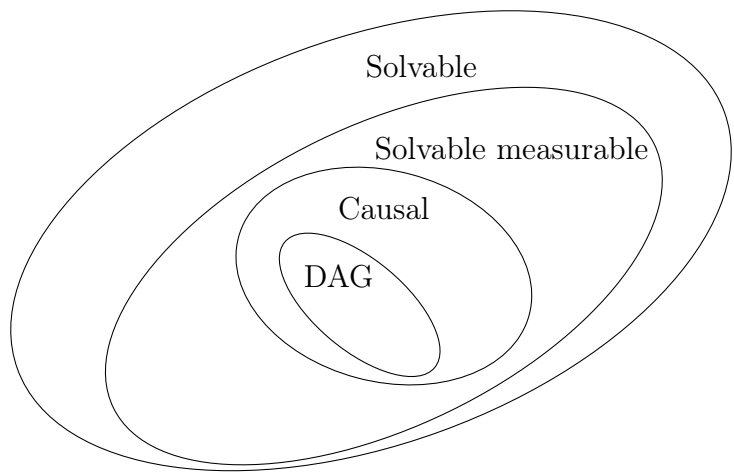
\begin{figure}[h]
  \centering
  \begin{tikzpicture}
    \draw[rotate=20] (0,0.5) ellipse (5cm and 2.7cm);	
    \draw[rotate=26] (0,0) ellipse (4cm and 2cm);	
    \draw[rotate=-15] (0,0) ellipse (2cm and 1.4cm);	
    \draw[rotate=50] (-0.5,0) ellipse (0.6cm and 1.3cm);	
    \node at (-0.6,0) {DAG};
    \node at (0,0.9) {Causal};
    \node at (1.7,1.7) {Solvable measurable};
    \node at (0.5,2.8) {Solvable};
  \end{tikzpicture}
  \caption{Hierarchy of systems\label{convergence1}}
\end{figure}

\section{Formal definitions and proofs}
\label{sec:formal-proof}

In this Sect.~\ref{sec:formal-proof}, we provide formal definitions and proofs.
In~\S\ref{subsec:conditonal-parentality},
we define conditional parentality and topological separation.
We prove that topological separation implies factorization
in~\S\ref{Topological_separation_implies_factorization}.
In~\S\ref{Nonrecursive_solvability_implies_conditional_independence},
we provides tools to study conditional independence in
the presence of nonrecursive systems.
In~\S\ref{subseq:discrete-continuous},
we discuss the impact of the discrete \emph{versus} continuous settings on
conditional independence.
In~\S\ref{Topological_separation_implies_conditional_independence},
we show that topological separation implies conditional independence.
In~\S\ref{Topological_separation_implies_the_do-calculus},
we show that the topological separation allows to define an
alternative (and equivalent) do-calculus.
\medskip

We provide background on binary relations.
We recall that a \emph{(binary) relation}~$\relation$ on~$\AGENT$ is
a subset $\relation \subset \AGENT^2 $ and that 
\( \agent\, \relation\, \bgent \) means 
\( \np{\agent,\bgent} \in \relation \).
For any subset \( \Bgent \subset \AGENT \), 
the \emph{(sub)diagonal relation} is \( \Delta_{\Bgent} = \bset{ \np{\agent,\bgent} \in \AGENT^2 }%
{ \agent=\bgent \in \Bgent } \)
and the \emph{diagonal relation} is \( \Delta=\Delta_{\AGENT} \).
A \emph{foreset} of a relation~$\relation$ is
any set of the form \( \relation \, \bgent = 
\defset{ \agent \in  \AGENT }{ \agent\, \relation \, \bgent } \),
where \( \bgent \in \AGENT \), 
or, by extension, of the form \( \relation \, \Bgent = 
\defset{ \agent \in  \AGENT }{ \exists \bgent \in \Bgent \eqsepv \agent\,
  \relation \, \bgent } \), where \( \Bgent \subset \AGENT \).
The \emph{opposite} or \emph{complementary~$\Complementary{\relation}$} of a binary
relation~$\relation$ is the relation~$\Complementary{\relation}=\AGENT^2\setminus\relation$,
that is, defined by \( \agent\, \relation^{\mathsf{c}} \, \bgent \iff 
\neg \np{ \agent\, \relation \, \bgent } \).
The \emph{converse~$\Converse{\relation}$} of a binary relation~$\relation$ is
defined by \( \agent\, \Converse{\relation} \, \bgent \iff \bgent\, \relation \, \agent
\) (and $\relation$ is  \emph{symmetric} if \( \Converse{\relation}=\relation \)).
The \emph{composition} $\relation\relation'$ of two
binary relations~$\relation, \relation'$ on~$\AGENT$ is defined by
\( \agent (\relation\relation') \bgent \iff
\exists \delta \in  \AGENT \), \( \agent\, \relation \, \delta \) and \( \delta\, \relation'
\, \bgent \);
then, by induction we define
\( \relation^{k+1}=\relation\relation^{k} \) for \( k \in \NN^* \). 
The \emph{transitive closure} of a binary relation~$\relation$ is
\( \TransitiveClosure{\relation} = \cup_{k=1}^{\infty} \relation^{k} \)
(and $\relation$ is  \emph{transitive} if \( \TransitiveClosure{\relation}=\relation \))
and the \emph{reflexive and transitive closure} is 
\( \TransitiveReflexiveClosure{\relation}= \TransitiveClosure{\relation} \cup \Delta \).

\subsection{Conditional parentality and topological separation}
\label{subsec:conditonal-parentality}

We now formally define the conditional parental relation,
provide properties and then deduce a topology 
on the set~\( \AGENT \) of agents. It is this topology which is implicit in
the Definition~\ref{de:topologically_separated} of topological separation.

\begin{definition}
  \label{de:conditional_parental_relation}
  Let $\HistorySubset\subset \HISTORY$ be a subset of configurations,
  and $\AgentSubsetW\subset\AGENT$ be a subset of agents.
  We set
  \begin{subequations}
    \label{eq:conditional_parental_relation}
    \begin{equation}
      \PrecedenceWH \agent = 
      \bigcap_{\Bgent\subset\AGENT; \tribu{\Information}_{\agent}\cap \HistorySubset
 \subset \tribu{\History}_{\Bgent\cup \AgentSubsetW}\cap \HistorySubset}
\!\!\!\!\!\!\!\!\!\!\!\!\!\!\Bgent
   \eqsepv \forall \agent \in \AGENT
      \eqfinv
      \label{eq:conditional_parental_relation_a}
    \end{equation}
    and we define the 
    (conditional) \emph{parental relation}~\( \PrecedenceWH \) on~$\AGENT$
    (\wrt\ $\np{\AgentSubsetW,\HistorySubset}$) 
    by
    \begin{equation}
      \bgent \mathrel{\PrecedenceWH} \agent 
      \iff 
      \bgent \in \PrecedenceWH \agent 
      \eqsepv \forall \np{\agent,\bgent} \in \AGENT^2 
      \eqfinp 
    \end{equation}  
  \end{subequations}
  We call (conditional) \emph{ancestral relation}
  (\wrt\ $\np{\AgentSubsetW,\HistorySubset}$) 
  the transitive and reflexive closure~$\ConditionalAncestor$  of 
  the conditional parental relation~\( \PrecedenceWH \),
  that is,
  \begin{equation}
    \ConditionalAncestor=
    \Delta\cup  \TransitiveClosure{\PrecedenceWH}
    = 
    \Delta\cup \bigcup_{k=1}^{\infty} \PrecedenceWH^{k} 
    \eqfinp
    \label{eq:ancestral_relation}
  \end{equation}
\end{definition}
Thus, when \( \bgent \mathrel{\PrecedenceWH} \agent \),
it means, by~\eqref{eq:conditional_parental_relation},
that the information available to agent~$\agent$,
on the subset~$\HistorySubset\subset \HISTORY$ of configurations, 
involves the decisions of the agent~$\bgent$ and, possibly 
of the agents in~$\AgentSubsetW$.
Witsenhausen's precedence relation~$\Precedence$
in~\eqref{eq:precedence_relation} is the special 
case~\( \PrecedenceEmptyHISTORY \).
\begin{proposition}
  We have that    
  \begin{equation}
    \PrecedenceWH = \Delta_{\Complementary{\AgentSubsetW}} \PrecedenceEmptyH 
    \eqfinp
    \label{eq:Precedence_PrecedenceWH}
  \end{equation}
\end{proposition}

\begin{proof}
  Let $\agent \in \AGENT$ be a given agent.
  We introduce the two subsets of agents defined by 
  $\Gamma_{\agent} =
  \bset{\Bgent \subset \AGENT}{ \tribu{\Information}_{\agent}\cap \HistorySubset
    \subset \tribu{\History}_{\Bgent}  \cap \HistorySubset}$ and 
  $\Gamma_{\agent,\AgentSubsetW}= \bset{\Bgent\subset\AGENT}{
    \tribu{\Information}_{\agent}
    \cap \HistorySubset \subset \tribu{\History}_{\Bgent\cup \AgentSubsetW}
  \cap \HistorySubset} $.
  Then, the two subsets $\PrecedenceEmptyH \agent$ and $\PrecedenceWH \agent$ read as 
  $\PrecedenceEmptyH \agent = \bigcap_{\Bgent \in \Gamma_{\agent}} \Bgent$
  and 
  $\PrecedenceWH \agent=\bigcap_{\Bgent \in \Gamma_{\agent,\AgentSubsetW}} \Bgent$.

  
  
  As a preliminary result we prove that,
  for any $\agent \in \AGENT$, we have that
  $\PrecedenceWH \agent \subset \Complementary{\AgentSubsetW}$.
  Indeed, for a given agent $\agent \in \AGENT$, we have that $\tribu{\Information}_{\agent}\subset \tribu{\History}
  = \tribu{\History}_{\Complementary{\AgentSubsetW} \cup
    \AgentSubsetW}$. Thus, we obtain that
  $\tribu{\Information}_{\agent} \cap \HistorySubset\subset
  \tribu{\History}_{\Complementary{\AgentSubsetW} \cup \AgentSubsetW}  \cap \HistorySubset$ which gives that 
  $\Complementary{\AgentSubsetW} \in \Gamma_{\agent,\AgentSubsetW}$.
  Now, as $\PrecedenceWH \agent=\bigcap_{\Bgent \in \Gamma_{\agent,\AgentSubsetW}} \Bgent$, we conclude that
  $\PrecedenceWH \agent \subset \Complementary{\AgentSubsetW}$.

  Now, we establish two easy inclusions.
  First, for any $\Cgent \in \Gamma_{\agent,\AgentSubsetW}$, using the definitions of $\Gamma_{\agent,\AgentSubsetW}$ and
  $\Gamma_{\agent}$ we obtain that $\Cgent \cup \AgentSubsetW \in \Gamma_{\agent}$. Thus, we have the inclusion
  $
  \bset{ \Cgent \cup {\AgentSubsetW}}{\Cgent \in \Gamma_{\agent,\AgentSubsetW}}
  \subset \Gamma_{\agent}
  $ which gives 
  \begin{equation}
    { \bigcap_{\Bgent \in \Gamma_{\agent}} \Bgent}
    \subset 
    { \bigcap_{ \Cgent \in \Gamma_{\agent,\AgentSubsetW}} \bp{\Cgent \cup
        \AgentSubsetW}}
    \eqfinp
    \label{included-in-gamma-a}
  \end{equation}
  Second, for any $\Bgent \in \Gamma_{\agent}$, we have that
  $\tribu{\Information}_{\agent}\cap \HistorySubset
  \subset \tribu{\History}_{\Bgent}  \cap \HistorySubset
  \subset \tribu{\History}_{\Bgent\cup\AgentSubsetW }  \cap \HistorySubset
  = \tribu{\History}_{(\Bgent\cap \Complementary{\AgentSubsetW}) \cup\AgentSubsetW }  \cap \HistorySubset$ which gives that
  $\Bgent\cap \Complementary{\AgentSubsetW} \in\Gamma_{\agent,\AgentSubsetW}$. Thus we have that
  $ 
  \bset{ \Bgent \cap \Complementary{\AgentSubsetW}}{\Bgent \in \Gamma_{\agent}}
  \subset \Gamma_{\agent,\AgentSubsetW}
  $, which gives
  \begin{equation}
    \bigcap_{ \Cgent \in \Gamma_{\agent,\AgentSubsetW}} \Cgent
    \subset 
    \bigcap_{\Bgent \in \Gamma_{\agent}}\bp{\Bgent \cap  \Complementary{\AgentSubsetW}}
    \eqfinp
    \label{included-in-gamma-a-W}
  \end{equation}

  Now, we successively have
  \begin{align*}
    \Complementary{\AgentSubsetW} \cap \bp{\PrecedenceEmptyH \agent}
    &= \Complementary{\AgentSubsetW} \cap \Bp{ \bigcap_{\Bgent \in \Gamma_{\agent}} \Bgent}
      \tag{as \( \PrecedenceEmptyH \agent = \bigcap_{\Bgent \in \Gamma_{\agent}}
      \Bgent \)}
    \\
    &\subset \Complementary{\AgentSubsetW} \cap
      \Bp{ \bigcap_{ \Cgent \in \Gamma_{\agent,\AgentSubsetW}} \bp{\Cgent \cup
      \AgentSubsetW}}
      \tag{by~\eqref{included-in-gamma-a}}
    \\
    &=
      \Complementary{\AgentSubsetW} \cap
      \Bp{ \bp{ \bigcap_{C \in \Gamma_{\agent,\AgentSubsetW}} \Cgent } \cup
      \AgentSubsetW} 
    \\
    &=
      \Bp{ \Complementary{\AgentSubsetW} \cap
      \bp{ \bigcap_{C \in \Gamma_{\agent,\AgentSubsetW}} \Cgent } }
      \cup \Bp{ \Complementary{\AgentSubsetW} \cap       \AgentSubsetW}
    \\
    &=
      \Complementary{\AgentSubsetW} \cap \bp{\PrecedenceWH \agent}
      \tag{as \(    \PrecedenceWH \agent=\bigcap_{\Cgent \in
      \Gamma_{\agent,\AgentSubsetW}} \Cgent \)}
    \\
    &=
      \PrecedenceWH \agent \tag{by preliminary result $\PrecedenceWH \agent \subset\Complementary{\AgentSubsetW}$}
    \\
    &\subset
      \bigcap_{\Bgent \in \Gamma_{\agent}}\bp{\Bgent \cap  \Complementary{\AgentSubsetW}}
      \tag{by~\eqref{included-in-gamma-a-W} as \(    \PrecedenceWH \agent=\bigcap_{\Cgent \in
      \Gamma_{\agent,\AgentSubsetW}} \Cgent  \)}
    \\
    &= \Complementary{\AgentSubsetW} \cap \bp{\PrecedenceEmptyH \agent}
      \eqfinp
  \end{align*}

  Therefore, we have obtained that
  \( \Complementary{\AgentSubsetW} \cap \bp{\PrecedenceEmptyH \agent} 
  = \PrecedenceWH \agent \), that is,
  \( \Delta_{\Complementary{\AgentSubsetW} } \PrecedenceEmptyH = \PrecedenceWH \). 
\end{proof}

In \citep{Witsenhausen:1975}, Witsenhausen 
introduced a topology on the set~\( \AGENT \) of agents
related to the precedence relation~$\Precedence$
in~\eqref{eq:precedence_relation}.
Here, we extend his approach to the conditional parental relation~\(
\PrecedenceWH \) on~$\AGENT$.

\begin{proposition}
  \label{pr:conditional_topology}
  There exists a topology on the set~\( \AGENT \) of agents such that 
  the topological closure~\( \TopologicalClosure{\Bgent} \) 
  of a subset~\( \Bgent \subset \AGENT\) is the $\ConditionalAncestor$-foreset
  \begin{equation}
    \TopologicalClosure{\Bgent} = \ConditionalAncestor\Bgent 
    \eqfinp
    \label{eq:ClosedTopology}
  \end{equation}
  In this topology, the subset~\( \AgentSubsetW \) is open or,
  equivalently, the subset~\( \Complementary{\AgentSubsetW} \) is closed.
\end{proposition}

  We refer the reader to the companion
paper~\cite{De-Lara-Chancelier-Heymann-2021}
for additional material on the aforementioned topology.
For the sake of completeness, we give a proof.
We mention that $\TopologicalClosure{\Bgent}$ is the smallest subset that
contains $\Bgent$ such that \(
\PrecedenceWH \TopologicalClosure{\Bgent} \subset  \TopologicalClosure{\Bgent}\).

\begin{proof}
  We define \( \ClosedTopologyWH = \bset{ \Cgent \subset \AGENT}{%
    \ConditionalAncestor\Cgent \subset \Cgent } \)
  and we show that the elements of~\( \ClosedTopologyWH \)
  form the closed sets of a topology on the set~\( \AGENT\).
  In fact, we are going to prove the stronger property that the set~\(
  \ClosedTopologyWH \) is an Alexandrov topology: it
  contains both \( \emptyset, \AGENT \) and is stable under 
  the union and intersection operations (not necessarily finite).
  Indeed, both \( \emptyset, \AGENT \in \ClosedTopologyWH \) as 
  \( \ConditionalAncestor\emptyset=\emptyset\) and
  \( \ConditionalAncestor\AGENT \subset \AGENT \).
  Moreover, let \( \sequence{\Bgent_\scenario}{\scenario\in\SCENARIO} \) 
  be family of elements \( \Bgent_\scenario \in \ClosedTopologyWH \), that is,
  \( \ConditionalAncestor\Bgent_\scenario \subset \Bgent_\scenario \),
  for all \( \scenario\in\SCENARIO \).
  We have \( \ConditionalAncestor\np{\bigcup_{\scenario\in\SCENARIO}\Bgent_\scenario }
  = \bigcup_{\scenario\in\SCENARIO}\ConditionalAncestor\Bgent_\scenario
  \subset \bigcup_{\scenario\in\SCENARIO}\Bgent_\scenario \), hence stability by union.
  We have \( \ConditionalAncestor\np{\bigcap_{\scenario\in\SCENARIO}\Bgent_\scenario }
  \subset \bigcap_{\scenario\in\SCENARIO}\ConditionalAncestor\Bgent_\scenario
  \subset \bigcap_{\scenario\in\SCENARIO}\Bgent_\scenario  \), hence stability by intersection.
  Thus, we have shown that \( \ClosedTopologyWH \) is an Alexandrov topology.

  We prove Equation~\eqref{eq:ClosedTopology}.
  By definition of the \( \ClosedTopologyWH \) topology, a subset~\( \Bgent \subset \AGENT\) is closed
  iff \( \ConditionalAncestor\Bgent \subset \Bgent\). 
  This is also equivalent to \( \ConditionalAncestor\Bgent = \Bgent\)
  because \( \Bgent \subset \ConditionalAncestor\Bgent \) since 
  the relation~\( \ConditionalAncestor = 
  \npTransitiveClosure{\PrecedenceWH}\cup\Delta \) is reflexive.
  Now, we consider a subset~\( \Bgent \subset \AGENT\) and we characterize 
  its topological closure~\( \TopologicalClosure{\Bgent} \), the smallest closed subset 
  that contains~\( \Bgent \). 
  On the one hand, we have that \( \Bgent \subset \ConditionalAncestor\Bgent \) 
  because the relation~\( \ConditionalAncestor = 
  \npTransitiveClosure{\PrecedenceWH}\cup\Delta \) is reflexive.
  On the other hand, the set \( \ConditionalAncestor\Bgent \) 
  is closed since 
  \( \ConditionalAncestor\np{\ConditionalAncestor\Bgent}=
  \np{\ConditionalAncestor}^2\Bgent
  \subset \ConditionalAncestor\Bgent\), because the relation~\( \ConditionalAncestor = 
  \npTransitiveClosure{\PrecedenceWH}\cup\Delta \) is transitive.
  By definition of the topological closure~\( \TopologicalClosure{\Bgent}\)
  we deduce that \( \TopologicalClosure{\Bgent} \subset \ConditionalAncestor\Bgent \).
  Now, we prove the reverse inclusion. Let \( \Cgent \subset \AGENT \) be a closed subset such that 
  \( \Bgent \subset \Cgent \),  we necessarily have that 
  \( \ConditionalAncestor\Bgent  \subset \ConditionalAncestor\Cgent \subset \Cgent\) and thus
  \( \ConditionalAncestor\Bgent \subset \Cgent\).
  Now, considering the special case where $C=\TopologicalClosure{\Bgent}$ which is a closed subset containing
  $\Bgent$ we obtain that \( \ConditionalAncestor\Bgent  \subset\TopologicalClosure{\Bgent} \).
  We conclude that \( \ConditionalAncestor\Bgent= \TopologicalClosure{\Bgent} \).
    
  The subset~\( \AgentSubsetW \) is open because its complementary
  set~\( \Complementary{\AgentSubsetW} \) satisfies
  \( \ConditionalAncestor\Complementary{\AgentSubsetW} =
  \npTransitiveClosure{\PrecedenceWH}\Complementary{\AgentSubsetW} \cup
  \Complementary{\AgentSubsetW} \subset \Complementary{\AgentSubsetW} \), as
  \( \PrecedenceWH\, \AGENT \subset \Complementary{\AgentSubsetW} \) because
  \( \PrecedenceWH = \Delta_{\Complementary{\AgentSubsetW}} \PrecedenceEmptyH \)
  by~\eqref{eq:Precedence_PrecedenceWH} and by definition of the subdiagonal
  relation~$\Delta_{\Complementary{\AgentSubsetW}}$.

  This ends the proof.
\end{proof}

\subsection{Topological separation implies factorization}
\label{Topological_separation_implies_factorization}

From now on, as we deal for the first time with probability, 
we consider a countable W-model (as in Definition~\ref{de:W-model}), that is,
a W-model where all sets
$\AGENT$, \( \sequence{\CONTROL_{\agent}}{\agent \in \AGENT}\)
and $\Omega$ are countable,
equipped with the complete $\sigma$-algebras. 
Moreover, we suppose that the set~$\Omega$ of states of Nature, and its field~\(
\tribu{\NatureField} \) have the following product form:
\begin{equation}
  \Omega = \prod_{\agent \in \AGENT} \Omega_{\agent}
  \eqsepv
  \tribu{\NatureField} = \bigotimes_{\agent \in \AGENT} \tribu{\NatureField}_{\agent}
  \eqfinp
  \label{eq:states_of_Nature_product}
\end{equation}
For any nonempty subset $\Bgent \subset \AGENT$ of agents, we denote
\begin{subequations}
  \begin{align}
    \Omega_{\Bgent}
    &= \prod_{\bgent \in \Bgent} \Omega_{\bgent} 
      \eqsepv
      \tribu{\NatureField}_{\Bgent} = \bigotimes_{\bgent \in \Bgent}\tribu{\NatureField}_{\bgent}
    \\      
    \history_\Bgent 
    &=
      \sequence{\history_\bgent}{\bgent \in \Bgent}
      \in \prod \limits_{\bgent \in \Bgent} \CONTROL_\bgent
      \eqsepv \forall \history \in \HISTORY=\produit{\Omega}{\prod\limits_{\agent \in \AGENT} \CONTROL_\agent}
      \eqfinv
      \label{eq:sub_history_Bgent}
    \\
    \projection_\Bgent &:
                         \HISTORY \to \prod\limits_{\bgent \in \Bgent} \CONTROL_\bgent
                         \eqsepv
                         \history \mapsto \history_\Bgent
                         \eqfinv
                         \label{eq:projection_Bgent}
    \\    
    \wstrategy_\Bgent 
    &=
      \sequence{\wstrategy_\bgent}{\bgent \in \Bgent}
      \in \prod \limits_{\bgent \in \Bgent} \WSTRATEGY_{\bgent} 
      \eqsepv \forall \wstrategy \in \WSTRATEGY
      \eqfinp 
      \label{eq:sub_wstrategy_Bgent}
  \end{align}
  %
\end{subequations}

We are now going to show how conditional topological separation induces a factorization 
of the solution map (see Definition~\ref{de:solvability}). 
\begin{lemma}
  \label{th:decomposition}
  We consider a solvable countable W-model, 
  where the set~$\Omega$ of states of Nature has the product
  form~\eqref{eq:states_of_Nature_product}, 
  where each information field~\( \tribu{\Information}_{\agent} \)
  in~\eqref{eq:information_field_agent} is such that 
  \begin{equation}
    \tribu{\Information}_{\agent} \subset \produit{ \tribu{\NatureField}_{\agent} \otimes
    \bigotimes\limits_{\bgent \neq \agent} \na{ \emptyset, \Omega_\bgent}
}{
    \bigotimes \limits_{\cgent \in \AGENT} \tribu{\Control}_{\cgent} }
    \eqsepv
    \forall \agent \in \AGENT
    \eqfinp 
    \label{eq:agents-local-noises}
  \end{equation} 
  We also consider a policy profile $\wstrategy = \sequence{\wstrategy_{\agent}}{\agent \in \AGENT}
  \in \prod_{\agent \in \AGENT} \WSTRATEGY_{\agent} $,
  a subset~$\HistorySubset\subset \HISTORY$ of configurations, 
  and $\AgentSubsetY$, $\AgentSubsetW$ and $\AgentSubsetZ$ 
  three subsets of $\AGENT$, two by two disjoint and such that
  (see Definition~\ref{de:topologically_separated} for the notation~\(
  \ConditionalTopologicalSeparation \))
  \begin{equation}
    \AgentSubsetY \ConditionalTopologicalSeparation \AgentSubsetZ \mid \np{\AgentSubsetW,\HistorySubset} 
    \eqfinp 
  \end{equation}
  Then, there exist five subsets 
  $\AgentSubsetY', \AgentSubsetZ'$,
  $\AgentSubsetW_\AgentSubsetY,\AgentSubsetW_\AgentSubsetZ$,
  $\residual$ $\subset \AGENT$
  such that 
  \begin{subequations}
    \begin{equation}
      \AGENT =  \widetilde{\AgentSubsetY} \sqcup \widetilde{\AgentSubsetZ}\sqcup
      \residual
      \quad \text{where} \quad
      \widetilde{\AgentSubsetY} = \AgentSubsetY \sqcup \AgentSubsetY'\sqcup
      \AgentSubsetW_\AgentSubsetY
      \eqsepv
      \widetilde{\AgentSubsetZ} = \AgentSubsetZ \sqcup \AgentSubsetZ'\sqcup
      \AgentSubsetW_\AgentSubsetZ
      \eqsepv
      \AgentSubsetW  = \AgentSubsetW_\AgentSubsetY\sqcup
      \AgentSubsetW_\AgentSubsetZ
      \eqfinv 
      \label{eq:five_subsets_proof}
    \end{equation}
    and there exist three measurable mappings (reduced solution maps)
    \begin{equation}
      \partialReducedSolutionMap_{\wstrategy_{\widetilde{\AgentSubsetY}}} :
      \Omega_{\widetilde{\AgentSubsetY}} 
      \times
      \CONTROL_{\AgentSubsetW_\AgentSubsetZ} 
      \to \CONTROL_{\widetilde{\AgentSubsetY}}
      \eqsepv 
      \partialReducedSolutionMap_{\wstrategy_{\widetilde{\AgentSubsetZ}}} :
      \Omega_{\widetilde{\AgentSubsetZ}} 
      \times
      \CONTROL_{\AgentSubsetW_\AgentSubsetY}
      \to \CONTROL_{\widetilde{\AgentSubsetZ}}
      \eqsepv 
      \partialReducedSolutionMap_{\wstrategy_{\residual}} :
      \Omega_{\residual} \times
      \CONTROL_{\widetilde{\AgentSubsetY}\cup\widetilde{\AgentSubsetZ}}
      \to \CONTROL_{\residual}
      \label{eq:three_mappings}
    \end{equation}
    such that the solution map $\SolutionMap_\wstrategy\np{\omega} = \bp{\omega,\ReducedSolutionMap_{\wstrategy}\np{\omega}}$
    in~\eqref{eq:solution_map} splits in three factors
    as follows:
    \( \forall \omega \in \SolutionMap_{\wstrategy}^{-1}(\HistorySubset) \), we have
    that 
    \begin{equation}
      \ReducedSolutionMap_\wstrategy\np{\omega} = \bgp{
        \partialReducedSolutionMap_{\wstrategy_{\widetilde{\AgentSubsetY}}}
        \Bp{ \omega_{\widetilde{\AgentSubsetY}}, 
          \wstrategy_{\AgentSubsetW_\AgentSubsetZ}\bp{\SolutionMap_{\wstrategy}\np{\omega}} },
        \partialReducedSolutionMap_{\wstrategy_{\widetilde{\AgentSubsetZ}}}
        \Bp{ \omega_{\widetilde{\AgentSubsetZ}}, 
          \wstrategy_{\AgentSubsetW_\AgentSubsetY}\bp{\SolutionMap_{\wstrategy}\np{\omega}} },
        \partialReducedSolutionMap_{\wstrategy_{\residual}}
        \Bp{ \omega_{\residual}, 
          \wstrategy_{\widetilde{\AgentSubsetY}\cup\widetilde{\AgentSubsetZ}}\bp{\SolutionMap_{\wstrategy}\np{\omega}} }
      }
      \eqfinv
      \label{eq:factorization}
    \end{equation}
    or, equivalently, with the notation~\eqref{eq:projection_Bgent},
    \begin{equation}
      \ReducedSolutionMap_\wstrategy\np{\omega} = \bgp{
        \partialReducedSolutionMap_{\wstrategy_{\widetilde{\AgentSubsetY}}}
        \Bp{ \omega_{\widetilde{\AgentSubsetY}}, 
          \projection_{\AgentSubsetW_\AgentSubsetZ}\bp{\SolutionMap_{\wstrategy}\np{\omega}} },
        \partialReducedSolutionMap_{\wstrategy_{\widetilde{\AgentSubsetZ}}}
        \Bp{ \omega_{\widetilde{\AgentSubsetZ}}, 
          \projection_{\AgentSubsetW_\AgentSubsetY}\bp{\SolutionMap_{\wstrategy}\np{\omega}} },
        \partialReducedSolutionMap_{\wstrategy_{\residual}}
        \Bp{ \omega_{\residual}, 
          \projection_{\widetilde{\AgentSubsetY}\cup\widetilde{\AgentSubsetZ}}\bp{\SolutionMap_{\wstrategy}\np{\omega}} }
      } \eqfinp 
      \label{eq:factorization_bis}
    \end{equation}
  \end{subequations}
More precisely, again with the notation~\eqref{eq:projection_Bgent},
Equation~\eqref{eq:factorization} has to be understood as 
\( \projection_{\widetilde{\AgentSubsetY}}\bp{ \ReducedSolutionMap_\wstrategy\np{\omega} } \)
\( = \) \( \partialReducedSolutionMap_{\wstrategy_{\widetilde{\AgentSubsetY}}}
\Bp{ \omega_{\widetilde{\AgentSubsetY}}, 
  \wstrategy_{\AgentSubsetW_\AgentSubsetZ}\bp{\SolutionMap_{\wstrategy}\np{\omega}}} \),
\( \projection_{\widetilde{\AgentSubsetZ}}\bp{ \ReducedSolutionMap_\wstrategy\np{\omega} } \)
\( = \) \( \partialReducedSolutionMap_{\wstrategy_{\widetilde{\AgentSubsetZ}}}
\Bp{ \omega_{\widetilde{\AgentSubsetZ}}, 
  \wstrategy_{\AgentSubsetW_\AgentSubsetY}\bp{\SolutionMap_{\wstrategy}\np{\omega}}} \), 
and
\( \projection_{\residual}\bp{ \ReducedSolutionMap_\wstrategy\np{\omega} } \)
\( = \)\\
\( \partialReducedSolutionMap_{\wstrategy_{\residual}}
\Bp{ \omega_{\residual}, 
  \wstrategy_{\widetilde{\AgentSubsetY}\cup\widetilde{\AgentSubsetZ}}
  \bp{\SolutionMap_{\wstrategy}\np{\omega}} } \). 
\end{lemma}

Before providing the proof we  introduce the  following preliminary result, which  is an application of a result by
Doob (see~\cite[Theorem 18 in
Chapter 2]{Dellacherie-Meyer:1975} and~\cite{doob1953stochastic}) in a
countable setting.
\begin{lemma}
  \label{lemma:doob}
We consider a countable W-model.  Let $\Agent\subset\AGENT$ and $\Bgent\subset\AGENT$.
  Let $\projection$ be the projection mapping from~$\HISTORY$ to
  $\produit{\Omega_\Agent}{ \CONTROL_\Bgent}$, and
\(   \policy_\Agent\)a policy profile for
  the elements of $\Agent$.
  Let $\History\subset\HISTORY$ be such that 
  \(\sigma(\policy_\Agent)\cap \History\subseteq\sigma(\projection)\cap H\). Then, there
  exist a mapping 
  $\hat{\policy}_\Agent:
  \produit{\Omega_\Agent}{ \CONTROL_\Bgent}\to \CONTROL_\Agent$ such that
  \(\lambda_\Agent^\History = \hat{\lambda}\circ \projection^\History\), where
  $\policy^\History_\Agent$  and $ \projection^\History$ are the restrictions of
  $\policy_\Agent$ and $\projection$ to~$\History$.
\end{lemma}
\begin{proof}
  This is a trivial application of Doob Lemma on a discrete set. 
  First we observe that the hypothesis implies that
  \(  \sigma(\policy_\Agent^\History)\subset \sigma(\projection^\History) \).
  As the set~\( \HISTORY \) is countable, so is the set~$\History$
  and we can apply Doob's Theorem~\cite[Theorem 18 in Chapter 2]{Dellacherie-Meyer:1975},
  which implies that there exists a mapping
  \(\hat{\policy} : \mathrm{im}(\pi^\History)\subset
  \produit{\Omega_\Agent}{ \CONTROL_\Bgent}\to \CONTROL_\Agent\) such that  \(\lambda_\Agent^\History =
  \hat{\lambda}\circ \projection^\History\)
  (where $\mathrm{im}(\pi^\History)$ is the the image of $\pi^\History$).
  We can then  extend the domain of $\hat{\policy}$, so that \(\hat{\policy} : 
  \produit{\Omega_\Agent}{ \CONTROL_\Bgent}\to \CONTROL_\Agent\) is such that  \(\lambda_\Agent^\History =
\hat{\lambda}\circ \projection^\History\)   which is what we wanted to
show. 
\end{proof}

\begin{proof} 
  The proof is in five steps.
  \medskip
  
  \begin{subequations}
    \noindent$\bullet$
    First, we identify five subsets 
    $\AgentSubsetY', \AgentSubsetZ'$,
    $\AgentSubsetW_\AgentSubsetY,\AgentSubsetW_\AgentSubsetZ$,
    $\residual$ $\subset \AGENT$ such that~\eqref{eq:five_subsets_proof} holds true.

    By assumption, we have that 
    \( \AgentSubsetY \cap \AgentSubsetZ = \AgentSubsetY \cap \AgentSubsetW
    = \AgentSubsetZ \cap \AgentSubsetW = \emptyset \) and 
    \( \AgentSubsetY \ConditionalTopologicalSeparation \AgentSubsetZ \mid \np{\AgentSubsetW,\HistorySubset} \).
    As a consequence, by Definition~\ref{de:topologically_separated},
    there exists 
    \( \AgentSubsetW_\AgentSubsetY , \AgentSubsetW_\AgentSubsetZ \subset \AgentSubsetW \) 
    such that 
    \( \AgentSubsetW_\AgentSubsetY \sqcup \AgentSubsetW_\AgentSubsetZ =
    \AgentSubsetW \) and 
    \( \TopologicalClosure{ \AgentSubsetY \cup \AgentSubsetW_\AgentSubsetY }
    \cap 
    \TopologicalClosure{ \AgentSubsetZ \cup \AgentSubsetW_\AgentSubsetZ }
    = \emptyset \), that is,
    \begin{equation}
      \AgentSubsetW_\AgentSubsetY \cap \AgentSubsetW_\AgentSubsetZ = \emptyset
      \eqsepv 
      \AgentSubsetW_\AgentSubsetY \cup \AgentSubsetW_\AgentSubsetZ = \AgentSubsetW 
      \eqsepv 
      \ConditionalAncestor\np{ \AgentSubsetY \cup \AgentSubsetW_\AgentSubsetY }
      \cap 
      \ConditionalAncestor\np{ \AgentSubsetZ \cup \AgentSubsetW_\AgentSubsetZ }
      = \emptyset 
      \eqfinp
      \label{eq:decomposition_assumptions_inproof}
    \end{equation}
    We set 
    \begin{equation}
      \widetilde{\AgentSubsetY}= \ConditionalAncestor
      \np{\AgentSubsetY \cup \AgentSubsetW_\AgentSubsetY}
      \, \text{ and }   \, 
      \widetilde{\AgentSubsetZ} = \ConditionalAncestor
      \np{\AgentSubsetZ\cup  \AgentSubsetW_\AgentSubsetZ}
      \eqfinp
      \label{eq:widetildeAgentSubsetY_inproof}
    \end{equation}
    By definition of the ancestral relation \( \ConditionalAncestor=  
    \TransitiveClosure{\PrecedenceWH}\cup\Delta \) 
    in~\eqref{eq:ancestral_relation}, we have that 
    \( \AgentSubsetY \cup \AgentSubsetW_\AgentSubsetY \subset
    \widetilde{\AgentSubsetY} \) and
    \( \AgentSubsetZ \cup \AgentSubsetW_\AgentSubsetZ\subset
    \widetilde{\AgentSubsetZ} \), where we can write
    \( \AgentSubsetY \cup \AgentSubsetW_\AgentSubsetY =
    \AgentSubsetY \sqcup \AgentSubsetW_\AgentSubsetY \) 
    and
    \( \AgentSubsetZ \cup \AgentSubsetW_\AgentSubsetZ =
    \AgentSubsetZ \sqcup \AgentSubsetW_\AgentSubsetZ \)
    because 
    \( \AgentSubsetY \cap \AgentSubsetW_\AgentSubsetY 
    \subset \AgentSubsetY \cap \AgentSubsetW =\emptyset \)
    and
    \( \AgentSubsetZ \cap \AgentSubsetW_\AgentSubsetZ 
    \subset \AgentSubsetZ \cap \AgentSubsetW =\emptyset \)
    by assumption. Then, we set 
    \[
      \AgentSubsetY'=
      \widetilde{\AgentSubsetY} \setminus 
      \bp{ \AgentSubsetY \sqcup \AgentSubsetW_\AgentSubsetY }
      \eqsepv 
      \AgentSubsetZ'=
      \widetilde{\AgentSubsetZ} \setminus 
      \bp{ \AgentSubsetZ \sqcup \AgentSubsetW_\AgentSubsetZ }
      \eqfinp
    \]
    \smallskip

    \noindent$\bullet$
    Second, we show that
   
    \begin{equation}
      \tribu{\Information}_{\widetilde{\AgentSubsetY}}\cap \HistorySubset
      \subset
      \tribu{\History}_{\widetilde{\AgentSubsetY}\cup
        \AgentSubsetW_\AgentSubsetZ} \cap \HistorySubset
      \, \text{ and }   \, 
      \tribu{\Information}_{\widetilde{\AgentSubsetZ}}\cap \HistorySubset
      \subset
      \tribu{\History}_{\widetilde{\AgentSubsetZ}\cup \AgentSubsetW_\AgentSubsetY}  \cap \HistorySubset
      \eqfinp 
      \label{eq:decomposition_proof}
    \end{equation}
    Indeed, we have that 
    \begin{align*}
      \PrecedenceWH \widetilde{\AgentSubsetY} 
      &=
        \PrecedenceWH 
        \ConditionalAncestor\np{\AgentSubsetY \cup \AgentSubsetW_\AgentSubsetY}
        \tag{as $\widetilde{\AgentSubsetY}=
        \ConditionalAncestor\np{\AgentSubsetY \cup \AgentSubsetW_\AgentSubsetY}$
        by definition~\eqref{eq:widetildeAgentSubsetY_inproof}}
      \\
      & \subset 
        \ConditionalAncestor\np{\AgentSubsetY \cup \AgentSubsetW_\AgentSubsetY}
        =   \widetilde{\AgentSubsetY}
        \tag{by definition~\eqref{eq:widetildeAgentSubsetY_inproof}}
        \eqfinp
    \end{align*}
    Therefore, using the fact that $\PrecedenceWH A \subset \Bgent  \iff
      \tribu{\Information}_{A}\cap \HistorySubset \subset
      \tribu{\History}_{\Bgent \cup \AgentSubsetW} \cap \HistorySubset$
      by definition of \(  \PrecedenceWH\), we get that
    \begin{equation*}
      \tribu{\Information}_{\widetilde{\AgentSubsetY}}\cap \HistorySubset
      \subset
      \tribu{\History}_{\widetilde{\AgentSubsetY}\cup
        \AgentSubsetW} \cap \HistorySubset
      \eqfinv
    \end{equation*}
    which, combined with the equality
    \(\widetilde{\AgentSubsetY}\cup
        \AgentSubsetW =\ConditionalAncestor(\AgentSubsetY\cup W_Y)\cup
        \AgentSubsetW_Y\cup\AgentSubsetW_Z=\ConditionalAncestor(\AgentSubsetY\cup W_Y)\cup
        \AgentSubsetW_Z= \widetilde{\AgentSubsetY}\cup
        \AgentSubsetW_Z\) gives
    \begin{equation*}
      \tribu{\Information}_{\widetilde{\AgentSubsetY}}\cap \HistorySubset
      \subset
      \tribu{\History}_{\widetilde{\AgentSubsetY}\cup
        \AgentSubsetW_\AgentSubsetZ} \cap \HistorySubset
      \eqfinp
    \end{equation*}
    In the same way, we obtain that 
    \(\tribu{\Information}_{\widetilde{\AgentSubsetZ}}\cap \HistorySubset
      \subset
      \tribu{\History}_{\widetilde{\AgentSubsetZ}\cup \AgentSubsetW_\AgentSubsetY}  \cap \HistorySubset\).
    \medskip
  \end{subequations}

  \noindent$\bullet$
  Third, we prepare the existence of a factorization as in~\eqref{eq:factorization}.
Using Equations~\eqref{eq:decomposition_assumptions_inproof} and~\eqref{eq:widetildeAgentSubsetY_inproof}
  we have that $\widetilde{\AgentSubsetY} \cap \widetilde{\AgentSubsetZ} = \emptyset$. We define
  $\residual = \AGENT\backslash \bp{\widetilde{\AgentSubsetY} \sqcup \widetilde{\AgentSubsetZ}}$
  to obtain the decomposition \( \AGENT =  \widetilde{\AgentSubsetY} \sqcup \widetilde{\AgentSubsetZ}\sqcup
  \residual \). The solution map~\eqref{eq:solution_map} splits in three factors
  (where the projection~$\pi$ has been introduced in~\eqref{eq:projection_Bgent})
  \[
    \ReducedSolutionMap_\wstrategy\np{\omega} = \Bp{ 
      \projection_{\widetilde{\AgentSubsetY}}\bp{ \SolutionMap_\wstrategy\np{\omega} },
      \projection_{\widetilde{\AgentSubsetZ}}\bp{ \SolutionMap_\wstrategy\np{\omega} },
      \projection_{\residual}\bp{ \SolutionMap_\wstrategy\np{\omega} } } 
    \eqfinp
  \]
  Let us examine the term 
  \( \projection_{\widetilde{\AgentSubsetY}}\bp{ \SolutionMap_\wstrategy\np{\omega} }
  \),
  as the other two terms can be treated in the same way.
  On the one hand, by~\eqref{eq:agents-local-noises}, we have that 
  \( \tribu{\Information}_{\widetilde{\AgentSubsetY}} 
  \subset \oproduit{ \tribu{\NatureField}_{\widetilde{\AgentSubsetY}} 
  \otimes
  \bigotimes\limits_{\bgent \not\in {\widetilde{\AgentSubsetY}}} \na{ \emptyset, \Omega_\bgent}
  }{
  \bigotimes \limits_{\cgent \in \AGENT} \tribu{\Control}_{\cgent} } \).
  On the other hand, by~\eqref{eq:decomposition_proof}, we have that 
  \( \tribu{\Information}_{\widetilde{\AgentSubsetY}} \cap \History
  \subset 
  \tribu{\History}_{\widetilde{\AgentSubsetY}\cup \AgentSubsetW_\AgentSubsetZ}\cap\History =
   \oproduit{ \tribu{\NatureField}
  }{
  \bigotimes 
  \limits_{\bgent \in \widetilde{\AgentSubsetY}\cup \AgentSubsetW_\AgentSubsetZ} 
  \tribu{\Control}_{\bgent} 
  \otimes
  \bigotimes 
  \limits_{\cgent \not\in \widetilde{\AgentSubsetY}\cup \AgentSubsetW_\AgentSubsetZ} 
  \na{ \emptyset, \CONTROL_{\cgent} } }\cap\History \).
  Therefore, we deduce that 
  \begin{align*}
    \tribu{\Information}_{\widetilde{\AgentSubsetY}} \cap\History
    &\subset
      \Bigg(
      \Bp{ \oproduit{ \tribu{\NatureField}_{\widetilde{\AgentSubsetY}} 
      \otimes
      \bigotimes\limits_{\bgent \not\in {\widetilde{\AgentSubsetY}}} \na{ \emptyset, \Omega_\bgent}
      }{
      \bigotimes \limits_{\cgent \in \AGENT} \tribu{\Control}_{\cgent} } }
      \cap
      \Bp{ \oproduit{ \tribu{\NatureField}
      }{
      \bigotimes 
      \limits_{\bgent \in \widetilde{\AgentSubsetY}\cup \AgentSubsetW_\AgentSubsetZ} 
      \tribu{\Control}_{\bgent} 
      \otimes
      \bigotimes 
      \limits_{\cgent \not\in \widetilde{\AgentSubsetY}\cup \AgentSubsetW_\AgentSubsetZ} 
      \na{ \emptyset, \CONTROL_{\cgent} } } }\Bigg)\cap\History
    \\
    &=\Bigg(
     \oproduit{ \tribu{\NatureField}_{\widetilde{\AgentSubsetY}} 
      \otimes
      \bigotimes\limits_{\bgent \not\in {\widetilde{\AgentSubsetY}}} 
      \na{ \emptyset, \Omega_\bgent}
      }{
      \bigotimes 
      \limits_{\bgent \in \widetilde{\AgentSubsetY}\cup \AgentSubsetW_\AgentSubsetZ} 
      \tribu{\Control}_{\bgent} 
      \otimes
      \bigotimes 
      \limits_{\cgent \not\in \widetilde{\AgentSubsetY}\cup \AgentSubsetW_\AgentSubsetZ} 
      \na{ \emptyset, \CONTROL_{\cgent} } }\Bigg)\cap\History
      \eqfinp
  \end{align*}
 By Lemma~\ref{lemma:doob},
  there exists a ``reduced'' mapping~\(
  \overline{\wstrategy}_{\widetilde{\AgentSubsetY}} :
  \Omega_{\widetilde{\AgentSubsetY}} 
  \times
  \CONTROL_{\widetilde{\AgentSubsetY}}
  \times
  \CONTROL_{\AgentSubsetW_\AgentSubsetZ} 
  \to \CONTROL_{\widetilde{\AgentSubsetY}} \) such that
  \[
    \begin{split}
      \wstrategy_{\widetilde{\AgentSubsetY}}\bp{
        \omega_{\widetilde{\AgentSubsetY}},
        \omega_{\AGENT\setminus\widetilde{\AgentSubsetY}},
        \control_{\widetilde{\AgentSubsetY}},
        \control_{\AgentSubsetW_\AgentSubsetZ},
        \control_{\AGENT\setminus\np{\widetilde{\AgentSubsetY}\cup\AgentSubsetW_\AgentSubsetZ}}}
      =
      \overline{\wstrategy}_{\widetilde{\AgentSubsetY}}\bp{
        \omega_{\widetilde{\AgentSubsetY}},
        \control_{\widetilde{\AgentSubsetY}},
        \control_{\AgentSubsetW_\AgentSubsetZ} }
      \eqsepv \\
      \forall 
      \bp{
        \omega_{\widetilde{\AgentSubsetY}},
        \omega_{\AGENT\setminus\widetilde{\AgentSubsetY}},
        \control_{\widetilde{\AgentSubsetY}},
        \control_{\AgentSubsetW_\AgentSubsetZ},
        \control_{\AGENT\setminus\np{\widetilde{\AgentSubsetY}\cup\AgentSubsetW_\AgentSubsetZ}}}
      \in \History
      \eqfinp      
    \end{split}
  \]
  In the same way,
  there exists a ``reduced'' mapping~\(
  \overline{\wstrategy}_{\widetilde{\AgentSubsetZ}} :
  \Omega_{\widetilde{\AgentSubsetZ}} 
  \times
  \CONTROL_{\widetilde{\AgentSubsetZ}}
  \times
  \CONTROL_{\AgentSubsetW_\AgentSubsetY} 
  \to \CONTROL_{\widetilde{\AgentSubsetZ}} \) such that
  \[
    \begin{split}
      \wstrategy_{\widetilde{\AgentSubsetZ}}\bp{
        \omega_{\widetilde{\AgentSubsetZ}},
        \omega_{\AGENT\setminus\widetilde{\AgentSubsetZ}},
        \control_{\widetilde{\AgentSubsetZ}},
        \control_{\AgentSubsetW_\AgentSubsetY},
        \control_{\AGENT\setminus\np{\widetilde{\AgentSubsetY}\cup\AgentSubsetW_\AgentSubsetZ}}}
      =
      \overline{\wstrategy}_{\widetilde{\AgentSubsetZ}}\bp{
        \omega_{\widetilde{\AgentSubsetZ}},
        \control_{\widetilde{\AgentSubsetZ}},
        \control_{\AgentSubsetW_\AgentSubsetY} }
      \eqsepv
      \\      \forall 
      \bp{
        \omega_{\widetilde{\AgentSubsetZ}},
        \omega_{\AGENT\setminus\widetilde{\AgentSubsetZ}},
        \control_{\widetilde{\AgentSubsetZ}},
        \control_{\AgentSubsetW_\AgentSubsetY},
        \control_{\AGENT\setminus\np{\widetilde{\AgentSubsetZ}\cup\AgentSubsetW_\AgentSubsetY}}}
      \in \History
      \eqfinv
    \end{split}
  \]
  and a mapping
  \(
  \overline{\wstrategy}_{\residual} :
  \Omega_{\residual}
  \times
  \CONTROL_{\residual}
  \times
  \CONTROL_{\Complementary{\residual}} 
  \to \CONTROL_{\residual} \) 
  such that
  \begin{align*}
    \wstrategy_{\residual}\bp{
    \omega_{\residual},
    \omega_{\Complementary{\residual}},
    \control_{\residual},
    \control_{\Complementary{\residual}}}
    &=
      \overline{\wstrategy}_{\residual}\bp{
      \omega_{\residual},
      \control_{\residual},
      \control_{\Complementary{\residual} }}
      \eqsepv
      \forall 
      \bp{
      \omega_{\residual},
      \omega_{\Complementary{\residual}},
      \control_{\residual},
      \control_{\Complementary{\residual}}}
      \in \History
      \eqfinp
  \end{align*}
  By Definition~\ref{de:solvability} of 
  \( \SolutionMap_\wstrategy\np{\omega} \), we can regroup ---  for  \(
   \omega \in \SolutionMap_{\wstrategy}^{-1}(\HistorySubset)
  \)--- 
  the closed-loop equations~\eqref{eq:solution_map_IFF}
  in three parts as 
  \begin{align*}
    \projection_{\widetilde{\AgentSubsetY}}\bp{ \SolutionMap_\wstrategy\np{\omega} }
    &= 
      \overline{\wstrategy}_{\widetilde{\AgentSubsetY}}\Bp{
      \omega_{\widetilde{\AgentSubsetY}},
      \projection_{\widetilde{\AgentSubsetY}}\bp{ \SolutionMap_\wstrategy\np{\omega} }, 
      \projection_{\AgentSubsetW_\AgentSubsetZ}\bp{ \SolutionMap_\wstrategy\np{\omega} } }
      \eqfinv
    \\
    \projection_{\widetilde{\AgentSubsetZ}}\bp{ \SolutionMap_\wstrategy\np{\omega} }
    &= 
      \overline{\wstrategy}_{\widetilde{\AgentSubsetZ}}\Bp{
      \omega_{\widetilde{\AgentSubsetZ}},
      \projection_{\widetilde{\AgentSubsetZ}}\bp{ \SolutionMap_\wstrategy\np{\omega} }, 
      \projection_{\AgentSubsetW_\AgentSubsetY}\bp{ \SolutionMap_\wstrategy\np{\omega} } }
      \eqfinv
    \\
    \projection_{\residual}\bp{ \SolutionMap_\wstrategy\np{\omega} }
    &= 
      \overline{\wstrategy}_{\residual}\Bp{
      \omega_{\residual},
      \projection_{\residual}\bp{ \SolutionMap_\wstrategy\np{\omega} },
      \projection_{\widetilde{\AgentSubsetY}}\bp{ \SolutionMap_\wstrategy\np{\omega} }, 
      \projection_{\widetilde{\AgentSubsetZ}}\bp{ \SolutionMap_\wstrategy\np{\omega} }
      }
            \eqfinv
  \end{align*}
  
  so that ---  for  \(
   \omega \in \SolutionMap_{\wstrategy}^{-1}(\HistorySubset)
  \)---  the reduced closed-loop equations
  \begin{subequations}
    \begin{align}
      \control_{\widetilde{\AgentSubsetY}}
      &=
        \overline{\wstrategy}_{\widetilde{\AgentSubsetY}}\bp{
        \omega_{\widetilde{\AgentSubsetY}},
        \control_{\widetilde{\AgentSubsetY}},
        \control_{\AgentSubsetW_\AgentSubsetZ} }
              \eqfinv
        \label{eq:closedLoop1}
      \\
      \control_{\widetilde{\AgentSubsetZ}}
      &=
        \overline{\wstrategy}_{\widetilde{\AgentSubsetZ}}\bp{
        \omega_{\widetilde{\AgentSubsetZ}},
        \control_{\widetilde{\AgentSubsetZ}},
        \control_{\AgentSubsetW_\AgentSubsetY} }
              \eqfinv
        \label{eq:closedLoop2}
      \\
      \control_{\residual}          &=
                                      \overline{\wstrategy}_{\residual}\bp{
                                      \omega_{\residual},
                                      \control_{\residual},
                                      \control_{\widetilde{\AgentSubsetY}},
                                      \control_{\widetilde{\AgentSubsetZ}}
                                      }
                                            \eqfinv
                                      \label{eq:closedLoop3}
    \end{align}
    \label{eq:reduced_closed-loop_equations_proof} 
  \end{subequations}
  have (at least) the solution
  \( \np{ \control_{\widetilde{\AgentSubsetY}},
    \control_{\widetilde{\AgentSubsetZ}},
    \control_{\residual}
  } =
  \Bp{ \projection_{\widetilde{\AgentSubsetY}}\bp{ \SolutionMap_\wstrategy\np{\omega} },
    \projection_{\widetilde{\AgentSubsetZ}}\bp{ \SolutionMap_\wstrategy\np{\omega} } ,
    \projection_{\residual}\bp{ \SolutionMap_\wstrategy\np{\omega} }
  } \)
  when
  \(\control_{\AgentSubsetW} = \projection_{\AgentSubsetW}\bp{ \SolutionMap_\wstrategy\np{\omega} }\).
  \medskip

  \noindent$\bullet$
  Fourth, we show the existence of three mappings as in~\eqref{eq:three_mappings}.

  Let  \(
   \omega \in \SolutionMap_{\wstrategy}^{-1}(\HistorySubset)
  \). 
  We denote by ${\CONTROL}_{\widetilde{\AgentSubsetY}}\np{\omega}$
  the set of elements $\control_{\widetilde{\AgentSubsetY}}\in
  \CONTROL_{\widetilde{\AgentSubsetY}}$ such that there exists at least one 
  $  (\omega'_{\widetilde{\AgentSubsetZ}},\omega'_{\residual},
  \control_{\widetilde{\AgentSubsetZ}},\control_{\residual}) \in \Omega_{\widetilde{\AgentSubsetZ}}\times \Omega_{\residual}\times
  \CONTROL_{\widetilde{\AgentSubsetZ}} \times 
  \CONTROL_{\residual}$ that satisfies 
  \(
  (\control_{\widetilde{\AgentSubsetY}},
  \control_{\widetilde{\AgentSubsetZ}},
  \control_{\residual}
  )
  =
  \wstrategy\bp{(
    \omega_{\widetilde{\AgentSubsetY}},
    \omega'_{\widetilde{\AgentSubsetZ}},
    \omega'_{\residual}
    ),
    (\control_{\widetilde{\AgentSubsetY}},
    \control_{\widetilde{\AgentSubsetZ}},
    \control_{\residual})}\)
  and \(
  \projection_{\AgentSubsetW_\AgentSubsetZ}\bp{ \SolutionMap_\wstrategy\np{\omega} }
  =\projection_{\AgentSubsetW_\AgentSubsetZ}\bp{ \SolutionMap_\wstrategy (
    \omega_{\widetilde{\AgentSubsetY}},
    \omega'_{\widetilde{\AgentSubsetZ}},
    \omega'_{\residual})
  } 
  \).
  We are going to show that ${\CONTROL}_{\widetilde{\AgentSubsetY}}\np{\omega}$
  is a singleton. 
  For this purpose, we consider 
  \(( \control_{\widetilde{\AgentSubsetY}},\omega_{\widetilde{\AgentSubsetZ}},\omega_{\residual}, \control_{\widetilde{\AgentSubsetZ}},\control_{\residual}) \)
  and  \((
  \control'_{\widetilde{\AgentSubsetY}},\omega'_{\widetilde{\AgentSubsetZ}},\omega'_{\residual},
  \control'_{\widetilde{\AgentSubsetZ}},\control'_{\residual}) \)
  satisfying the two conditions that define the set~${\CONTROL}_{\widetilde{\AgentSubsetY}}\np{\omega}$.

  As we have, on the one hand, that 
  \begin{subequations}
    \begin{equation}
      (\control_{\widetilde{\AgentSubsetY}},
      \control_{\widetilde{\AgentSubsetZ}},
      \control_{\residual}
      )
      =
      \wstrategy\bp{(
        \omega_{\widetilde{\AgentSubsetY}},
        \omega_{\widetilde{\AgentSubsetZ}},
        \omega_{\residual}
        ),
        (\control_{\widetilde{\AgentSubsetY}},
        \control_{\widetilde{\AgentSubsetZ}},
        \control_{\residual})}
    \end{equation}
    and, on the other hand, that
    \begin{equation}
      \label{eq:unicity1}
      (\control'_{\widetilde{\AgentSubsetY}},
      \control'_{\widetilde{\AgentSubsetZ}},
      \control'_{\residual}
      )
      =
      \wstrategy\bp{(
        \omega_{\widetilde{\AgentSubsetY}},
        \omega'_{\widetilde{\AgentSubsetZ}},
        \omega'_{\residual}
        ),
        (\control'_{\widetilde{\AgentSubsetY}},
        \control'_{\widetilde{\AgentSubsetZ}},
        \control'_{\residual})}
      \eqfinv
    \end{equation}
    we deduce that, by
    using~\eqref{eq:reduced_closed-loop_equations_proof} and the
    condition on \(\control_{\AgentSubsetW_\AgentSubsetZ}\) in the
    definition of ${\CONTROL}_{\widetilde{\AgentSubsetY}}\np{\omega}$:
    \begin{equation}
      \label{eq:unicity2}
      (\control_{\widetilde{\AgentSubsetY}},
      \control'_{\widetilde{\AgentSubsetZ}},
      \hat{\control}_{\residual}
      )
      =
      \wstrategy\bp{(
        \omega_{\widetilde{\AgentSubsetY}},
        \omega'_{\widetilde{\AgentSubsetZ}},
        \omega'_{\residual}
        ),
        (\control_{\widetilde{\AgentSubsetY}},
        \control'_{\widetilde{\AgentSubsetZ}},
        \hat{\control}_{\residual})}
      \eqfinv 
    \end{equation}
    where   \( 
    \hat{\control}_{\residual}=
    \overline{\wstrategy}_{\residual}\Bp{
      \omega'_{\residual},
      \hat{\control}_{\residual} ,
      \wstrategy_{\widetilde{\AgentSubsetY}\cup
        \widetilde{\AgentSubsetZ}}\bp{\SolutionMap_{\wstrategy}\np{\omega}}} \).
  \end{subequations}
  The  Equations  \eqref{eq:unicity1}  and  \eqref{eq:unicity2} imply,
  by the solvability assumption (see Definition~\ref{de:solvability}), that \( \control_{\widetilde{\AgentSubsetY}}
  = \control_{\widetilde{\AgentSubsetY}}' \).
  As a consequence, we have proven that the set~${\CONTROL}_{\widetilde{\AgentSubsetY}}\np{\omega}$
  is a singleton. 
  
  Thus, we have defined, for any 
  \(\control_{\AgentSubsetW_{\widetilde{\AgentSubsetZ}}} =
  \projection_{\AgentSubsetW_{\widetilde{\AgentSubsetZ}}} \bp{
    \SolutionMap_\wstrategy\np{\omega} }\)
  a unique element 
  \( \control_{\widetilde{\AgentSubsetY}}=
  \partialReducedSolutionMap_{\wstrategy_{\widetilde{\AgentSubsetY}}} 
  \np{\omega_{\widetilde{\AgentSubsetY}},
    \control_{\AgentSubsetW_{\widetilde{\AgentSubsetZ}}} } \).
  We do the same for \( \widetilde{\AgentSubsetZ} \) and for \( \residual \).
  Thus, we have defined reduced solution maps 
  as follows
  \begin{subequations}
    \begin{align}
      \partialReducedSolutionMap_{\wstrategy_{\widetilde{\AgentSubsetY}}} 
      &:
        \Omega_{\widetilde{\AgentSubsetY}} 
        \times
        \CONTROL_{\AgentSubsetW_\AgentSubsetZ} 
        \to \CONTROL_{\widetilde{\AgentSubsetY}}
        \eqfinv 
      \\
      \partialReducedSolutionMap_{\wstrategy_{\widetilde{\AgentSubsetZ}}} 
      &:
        \Omega_{\widetilde{\AgentSubsetZ}} 
        \times
        \CONTROL_{\AgentSubsetW_\AgentSubsetY} 
        \to \CONTROL_{\widetilde{\AgentSubsetZ}}
        \eqfinv 
      \\
      \partialReducedSolutionMap_{\wstrategy_{\residual}} 
      &:
        \Omega_{\residual} \times
        \CONTROL_{\widetilde{\AgentSubsetY} \cup \widetilde{\AgentSubsetZ}}
        \to \CONTROL_{\residual}
        \eqfinp
    \end{align}  
    \label{eq:reducedSolutionMapIntro}
  \end{subequations}
  As we considered that all sets
  $\AGENT$, \( \sequence{\CONTROL_{\agent}}{\agent \in \AGENT}\)
  and $\Omega$ are countable, the above mappings are measurable.
  \medskip

  This ends the proof.
\end{proof}

\subsection{Conditional independence in the presence of cycles}
\label{Nonrecursive_solvability_implies_conditional_independence}

This subsection provides tools to study conditional independence in
the presence of nonrecursive systems.
We also discuss an instance where such independence is not captured by 
Pearl's d-separation criterion~\cite{Pearl2011}.

\subsubsection{Key technical lemma for  dealing with cycles}

We state and prove a lemma that will serve as a main argument for the
proof of the coming
Theorem~\ref{th:push-forward-probability-decomposition}.
As far as we know, this result is novel. It cannot be deduced from
Pearl's rules.  

\begin{lemma}
  \label{lem:independence}

  Let $(\Omega,\tribu{\NatureField},\PROBA)$ be a probability space.
  Let $\Xi_1$, $\Xi_2$, $\Upsilon_1$, $\Upsilon_2$, $\Theta_1$, $\Theta_2$
  be six measurables spaces and 
  \begin{equation}
    \Psi_1 : \Xi_1 \times \Upsilon_2 \to \Theta_1
    \eqsepv
    \Psi_2 : \Xi_2 \times \Upsilon_1 \to \Theta_2
    \eqsepv
    \Phi_1 : \Xi_1 \times \Upsilon_2 \to \Upsilon_1
    \eqsepv
    \Phi_2 : \Xi_2 \times \Upsilon_1 \to \Upsilon_2 
  \end{equation}
  be four measurable mappings.
  Let  
  \begin{equation}
    \bm{\xi}_1 : \Omega \to \Xi_1
    \eqsepv
    \bm{\xi}_2 : \Omega \to \Xi_2
    \eqsepv
    \bm{\theta}_1  : \Omega \to \Theta_1
    \eqsepv
    \bm{\theta}_2  : \Omega \to \Theta_2
    \eqsepv
    \bm{\upsilon}_1  : \Omega \to \Upsilon_1 
    \eqsepv
    \bm{\upsilon}_2  : \Omega \to \Upsilon_2 
  \end{equation}
  be six random variables
  satisfying 
  \begin{subequations}
    \label{eq:six_random_variables}
    \begin{align}
      \bm{\theta}_1 &= \Psi_1\np{\bm{\xi}_1 ,\bm{\upsilon}_2} 
                      \eqfinv
                      \label{eq:six_random_variables_four_relations_theta_1}
      \\
      \bm{\theta}_2 &= \Psi_2\np{\bm{\xi}_2 ,\bm{\upsilon}_1} 
                      \eqfinv
                      \label{eq:six_random_variables_four_relations_theta_2}
      \\
      \bm{\upsilon}_1 &= \Phi_1\np{\bm{\xi}_1 ,\bm{\upsilon}_2} 
                        \eqfinv
                        \label{eq:six_random_variables_four_relations_upsilon_1}
      \\
      \bm{\upsilon}_2 &= \Phi_2\np{\bm{\xi}_2 ,\bm{\upsilon}_1}
                        \eqfinp
                        \label{eq:six_random_variables_four_relations_upsilon_2}
    \end{align}
  \end{subequations}
  Suppose that the couple $\np{\bm{\upsilon}_1,\bm{\upsilon}_2}$ of random variables takes 
  values in a countable
  product subset \( \Upsilon_1' \times \Upsilon_2' \subset \Upsilon_1 \times \Upsilon_2\),
  and that the system of equations
  \begin{subequations}
    \begin{align}
      w_1 &= \Phi_1\np{x_1 ,w_2}
            \eqfinv
      \\
      w_2 &= \Phi_2\np{x_2 ,w_1}
                        \eqfinv
    \end{align}
    \label{eq:unique_solution}
  \end{subequations}
  has a unique solution \( \np{w_1,w_2} \) in \( \Upsilon_1' \times \Upsilon_2' \),
  for any \( \np{x_1,x_2} \in \Xi_1 \times \Xi_2\).

  Then, if the random variables~$\bm{\xi}_1$ and $\bm{\xi}_2$ are independent,
  the random variables~$\bm{\theta}_1$ and $\bm{\theta}_2$
  are independent when conditioned on $(\bm{\upsilon}_1,\bm{\upsilon}_2)$.
\end{lemma}

\begin{proof}
  By assumption, there exists a unique solution \( \np{w_1,w_2} \in \Upsilon_1' \times \Upsilon_2' \) 
  to the implicit system~\eqref{eq:unique_solution} of equations. Thus, there
  exist mappings
  \begin{subequations}
    \begin{equation}
      \widetilde\Phi_1: \Xi_1 \times \Xi_2 \to \Upsilon_1'
      \eqsepv
      \widetilde\Phi_2 : \Xi_1 \times \Xi_2 \to \Upsilon_2' 
      \eqsepv
    \end{equation}
    such that, for any \( \np{w_1,w_2} \) in \( \Upsilon_1' \times \Upsilon_2' \)
    and \( \np{x_1,x_2} \in \Xi_1 \times \Xi_2\), we have 
    \begin{equation}
      \Bp{  w_1 = \Phi_1\np{x_1 ,w_2} \eqsepv w_2 = \Phi_2\np{x_2 ,w_1} }
      \iff 
      \Bp{  w_1 = \widetilde\Phi_1 (x_1 ,x_2) \eqsepv 
        w_2 = \widetilde\Phi_2 (x_1 ,x_2) }
      \eqfinp 
    \end{equation}
    \label{eq:unique_solution_mappings_inproof} 
  \end{subequations}
  We suppose that the random variables~$\bm{\xi}_1$ and $\bm{\xi}_2$ are
  independent.
  We are going to show that the random variables~$\bm{\theta}_1$ and $\bm{\theta}_2$
  are independent when conditioned on $(\bm{\upsilon}_1,\bm{\upsilon}_2)$.
  \smallskip

  \noindent$\bullet$
  First, we establish that, for any couple
  \( \np{w_1,w_2} \in \Upsilon_1' \times \Upsilon_2' \):
  \begin{equation}
    \Ba{ \Phi_1\np{\bm{\xi}_1,w_2}=w_1 \eqsepv 
      \Phi_2\np{\bm{\xi}_2,w_1}=w_2 }
    =
    \Ba{ \bm{\upsilon}_1 = w_1 \eqsepv \bm{\upsilon}_2 = w_2 }
    \eqfinp
    \label{eq:unique_solution_proof}
  \end{equation}
  Indeed,  on the one hand, we have
  \begin{align*}
    &\Ba{ \Phi_1\np{\bm{\xi}_1,w_2} =w_1 \eqsepv \Phi_2\np{\bm{\xi}_2,w_1}=w_2 }
    \\
    &\hspace{1cm}= 
      \Ba{ w_1 = \widetilde\Phi_1\np{\bm{\xi}_1,\bm{\xi}_2} \eqsepv 
      w_2 = \widetilde\Phi_2\np{\bm{\xi}_1,\bm{\xi}_2} }
      \tag{by~\eqref{eq:unique_solution_mappings_inproof}}
    \\
    &\hspace{1cm}= 
      \Ba{ w_1 = \widetilde\Phi_1\np{\bm{\xi}_1,\bm{\xi}_2} \eqsepv 
      w_2 = \widetilde\Phi_2\np{\bm{\xi}_1,\bm{\xi}_2} }
      \cap
      \Ba{ \bm{\upsilon}_1 = \Phi_1\np{\bm{\xi}_1 ,\bm{\upsilon}_2} \eqsepv 
      \bm{\upsilon}_2 = \Phi_2\np{\bm{\xi}_2 ,\bm{\upsilon}_1} }
      \intertext{because \( \Ba{ \bm{\upsilon}_1 = \Phi_1\np{\bm{\xi}_1,\bm{\upsilon}_2} \eqsepv 
      \bm{\upsilon}_2 = \Phi_2\np{\bm{\xi}_2 ,\bm{\upsilon}_1} } =\Omega \)
      by~\eqref{eq:six_random_variables_four_relations_upsilon_1} and \eqref{eq:six_random_variables_four_relations_upsilon_2}}
    &\hspace{1cm}=
      \Ba{ w_1 = \widetilde\Phi_1\np{\bm{\xi}_1,\bm{\xi}_2} \eqsepv 
      w_2 = \widetilde\Phi_2\np{\bm{\xi}_1,\bm{\xi}_2} }
      \cap
      \Ba{ \bm{\upsilon}_1 = \widetilde\Phi_1\np{\bm{\xi}_1,\bm{\xi}_2} \eqsepv 
      \bm{\upsilon}_2 = \widetilde\Phi_2\np{\bm{\xi}_1,\bm{\xi}_2} } 
      \tag{by~\eqref{eq:unique_solution_mappings_inproof}}
    \\
    &\hspace{1cm}\subset 
      \Ba{ \bm{\upsilon}_1 = w_1 \eqsepv \bm{\upsilon}_2 = w_2 }
      \eqfinp
  \end{align*}
  On the other hand, the reverse inclusion can be proved in the same way.
  Thus, we have obtained the equality~\eqref{eq:unique_solution_proof}.
  \smallskip

  \noindent$\bullet$
  Second, we show that, for any subsets \( \Theta_1' \subset \Theta_1 \) and
  \( \Theta_2' \subset \Theta_2 \), and for any couple
  \( \np{w_1,w_2} \in \Upsilon_1' \times \Upsilon_2' \), we have that
  \begin{align}
    &\Ba{ \bm{\theta}_1 \in \Theta_1' \eqsepv
      \bm{\theta}_2 \in \Theta_2' \eqsepv
      \bm{\upsilon}_1 = w_1 \eqsepv \bm{\upsilon}_2 = w_2 }
      \nonumber \\
    &\hspace{1cm}=
      \Ba{ \Psi_1\np{\bm{\xi}_1,w_2} \in \Theta_1' \eqsepv
      \Phi_1\np{\bm{\xi}_1,w_2}=w_1 } 
      \cap
      \Ba{ \Psi_2\np{\bm{\xi}_2,w_1} \in \Theta_2' \eqsepv
      \Phi_2\np{\bm{\xi}_2,w_1}=w_2 }
      \eqfinp
      \label{eq:intersection_of_two_sets_proof}
  \end{align}
  Indeed, we have that 
  \begin{subequations}
    \begin{align*}
      &
        \Ba{ \bm{\theta}_1 \in \Theta_1' \eqsepv
        \bm{\theta}_2 \in \Theta_2' \eqsepv
        \bm{\upsilon}_1 = w_1 \eqsepv \bm{\upsilon}_2 = w_2 }
      \\
      &\hspace{1cm}=
        \Ba{ \Psi_1\np{\bm{\xi}_1,\bm{\upsilon}_2} \in \Theta_1' \eqsepv
        \Psi_2\np{\bm{\xi}_2,\bm{\upsilon}_1} \in \Theta_2' \eqsepv
        \bm{\upsilon}_1 = w_1 \eqsepv \bm{\upsilon}_2 = w_2 }
        \tag{by~\eqref{eq:six_random_variables_four_relations_theta_1} 
        and \eqref{eq:six_random_variables_four_relations_theta_2}}
      \\
      &\hspace{1cm}=
        \Ba{ \Psi_1\np{\bm{\xi}_1,w_2} \in \Theta_1' \eqsepv
        \Psi_2\np{\bm{\xi}_2,w_1} \in \Theta_2' \eqsepv
        \bm{\upsilon}_1 = w_1 \eqsepv \bm{\upsilon}_2 = w_2 }
        \intertext{by substitution of the last two terms  \( \bm{\upsilon}_1 = w_1 \) and \( \bm{\upsilon}_2 = w_2 \)
        in the first two terms}
      &\hspace{1cm}=
        \left\{ \right. \Psi_1\np{\bm{\xi}_1,w_2} \in \Theta_1' \eqsepv
        \Psi_2\np{\bm{\xi}_2,w_1} \in \Theta_2' \eqsepv
        \Phi_1\np{\bm{\xi}_1,\bm{\upsilon}_2}=w_1 \eqsepv 
        \Phi_2\np{\bm{\xi}_2,\bm{\upsilon}_1}=w_2 \eqsepv
      \\
      & \hspace{2cm}\left.
        \bm{\upsilon}_1 = w_1 \eqsepv \bm{\upsilon}_2 = w_2 \right\}
        \tag{because \( \ba{ \Phi_1\np{\bm{\xi}_1,\bm{\upsilon}_2}=\bm{\upsilon}_1 \eqsepv 
        \Phi_2\np{\bm{\xi}_2,\bm{\upsilon}_1}=\bm{\upsilon}_2 } = \Omega \)
        by~\eqref{eq:six_random_variables_four_relations_upsilon_1}
        and \eqref{eq:six_random_variables_four_relations_upsilon_2}}
      \\
      &\hspace{1cm}=
        \left\{ \right. \Psi_1\np{\bm{\xi}_1,w_2} \in \Theta_1' \eqsepv
        \Psi_2\np{\bm{\xi}_2,w_1} \in \Theta_2' \eqsepv
        \Phi_1\np{\bm{\xi}_1,w_2}=w_1 \eqsepv 
        \Phi_2\np{\bm{\xi}_2,w_1}=w_2 \eqsepv 
      \\
      & \hspace{2cm}\left.
        \bm{\upsilon}_1 = w_1 \eqsepv \bm{\upsilon}_2 = w_2  \right\}
        \intertext{by substitution of the last two terms  \( \bm{\upsilon}_1 = w_1 \) and \( \bm{\upsilon}_2 = w_2 \)
        in the two middle terms}
      &\hspace{1cm}=
        \Ba{ \Psi_1\np{\bm{\xi}_1,w_2} \in \Theta_1' \eqsepv
        \Psi_2\np{\bm{\xi}_2,w_1} \in \Theta_2' \eqsepv
        \Phi_1\np{\bm{\xi}_1,w_2}=w_1 \eqsepv 
        \Phi_2\np{\bm{\xi}_2,w_1}=w_2 }
        \eqfinp
    \end{align*}
  \end{subequations}
  because the system~\eqref{eq:unique_solution} of equations 
  has a unique solution on \( \Upsilon_1' \times \Upsilon_2' \),
  so that \( \Phi_1\np{\bm{\xi}_1,w_2}=w_1 \) and \( \Phi_2\np{\bm{\xi}_2,w_1}=w_2 \)
  imply that \( \bm{\upsilon}_1 = w_1 \) and \( \bm{\upsilon}_2 = w_2 \) hold true
  by~\eqref{eq:six_random_variables_four_relations_upsilon_1}
  and \eqref{eq:six_random_variables_four_relations_upsilon_2}.
  Thus, we have obtained~\eqref{eq:intersection_of_two_sets_proof}.

  \smallskip

  \noindent$\bullet$
  Third, and finally, we show that the random variables~$\bm{\theta}_1$ and $\bm{\theta}_2$
  are independent when conditioned on $(\bm{\upsilon}_1,\bm{\upsilon}_2)$.
  For this purpose, we calculate, for any subsets \( \Theta_1' \subset \Theta_1 \) and
  \( \Theta_2' \subset \Theta_2 \), and for any couple
  \( \np{w_1,w_2} \in \Upsilon_1' \times \Upsilon_2' \):
  \begin{subequations}
    \begin{align*}
      &
        \PROBA\Bset{ \bm{\theta}_1 \in \Theta_1' \eqsepv \bm{\theta}_2 \in \Theta_2'}
        { \bm{\upsilon}_1 = w_1 \eqsepv \bm{\upsilon}_2 = w_2 }
      \\
      &\hspace{1cm}=
        \frac{ \PROBA\Ba{ \bm{\theta}_1 \in \Theta_1' \eqsepv
        \bm{\theta}_2 \in \Theta_2' \eqsepv
        \bm{\upsilon}_1 = w_1 \eqsepv \bm{\upsilon}_2 = w_2 } }%
        { \PROBA\Ba{ \bm{\upsilon}_1 = w_1 \eqsepv \bm{\upsilon}_2 = w_2 } }
        \intertext{by definition of the conditional probability, and where all
        quantities are zero if the denominator is zero}
      &\hspace{1cm}=
        \frac{ \PROBA\Ba{ \bm{\theta}_1 \in \Theta_1' \eqsepv
        \bm{\theta}_2 \in \Theta_2' \eqsepv
        \bm{\upsilon}_1 = w_1 \eqsepv \bm{\upsilon}_2 = w_2 } }%
        { \PROBA\Ba{ \bm{\theta}_1 \in \Theta_1 \eqsepv
        \bm{\theta}_2 \in \Theta_2 \eqsepv \bm{\upsilon}_1 = w_1 \eqsepv \bm{\upsilon}_2 = w_2 }
        }
        \tag{because \( \Ba{ \bm{\theta}_1 \in \Theta_1 \eqsepv
        \bm{\theta}_2 \in \Theta_2 } = \Omega \)}
      \\
      &\hspace{1cm}=
        \frac{ \PROBA      \Ba{ \Psi_1\np{\bm{\xi}_1,w_2} \in \Theta_1' \eqsepv
        \Phi_1\np{\bm{\xi}_1,w_2}=w_1 } 
        \times 
        \PROBA      \Ba{ \Psi_2\np{\bm{\xi}_2,w_1} \in \Theta_2' \eqsepv
        \Phi_2\np{\bm{\xi}_2,w_1}=w_2 } }%
        { \PROBA      \Ba{ \Psi_1\np{\bm{\xi}_1,w_2} \in \Theta_1 \eqsepv
        \Phi_1\np{\bm{\xi}_1,w_2}=w_1 } 
        \times 
        \PROBA      \Ba{ \Psi_2\np{\bm{\xi}_2,w_1} \in \Theta_2 \eqsepv
        \Phi_2\np{\bm{\xi}_2,w_1}=w_2 } }
        \intertext{by~\eqref{eq:intersection_of_two_sets_proof}, and then using the 
        assumption that the random variables~$\bm{\xi}_1$ and $\bm{\xi}_2$ are independent}
      &\hspace{1cm}=
        \frac{ \PROBA      \Ba{ \Psi_1\np{\bm{\xi}_1,w_2} \in \Theta_1' \eqsepv
        \Phi_1\np{\bm{\xi}_1,w_2}=w_1 } }%
        { \PROBA      \Ba{ \Phi_1\np{\bm{\xi}_1,w_2}=w_1 } }
        \times 
        \frac{ \PROBA      \Ba{ \Psi_2\np{\bm{\xi}_2,w_1} \in \Theta_2' \eqsepv
        \Phi_2\np{\bm{\xi}_2,w_1}=w_2 } }%
        { \PROBA    \Ba{ \Phi_2\np{\bm{\xi}_2,\bm{\upsilon}_1}=w_2 } }
               \tag{because \( \Ba{ \Psi_1\np{\bm{\xi}_1,w_2} \in \Theta_1 \eqsepv
        \Psi_2\np{\bm{\xi}_2,w_1} \in \Theta_2 } = \Omega \).}
    \end{align*}
  \end{subequations}
  Then, we focus on the first term of the product and we write
  \begin{align*}
    &
      \frac{ \PROBA      \Ba{ \Psi_1\np{\bm{\xi}_1,w_2} \in \Theta_1' \eqsepv
      \Phi_1\np{\bm{\xi}_1,w_2}=w_1 } }%
      { \PROBA      \Ba{ \Phi_1\np{\bm{\xi}_1,w_2}=w_1 } }
    \\
    &\hspace{1cm}=
      \frac{ \PROBA      \Ba{ \Psi_1\np{\bm{\xi}_1,w_2} \in \Theta_1' \eqsepv
      \Phi_1\np{\bm{\xi}_1,w_2}=w_1 } 
      \times \PROBA\Ba{ \Phi_2\np{\bm{\xi}_2,\bm{\upsilon}_1}=w_2 } }%
      { \PROBA\Ba{ \Phi_1\np{\bm{\xi}_1,w_2}=w_1 } \times \PROBA\Ba{ \Phi_2\np{\bm{\xi}_2,\bm{\upsilon}_1}=w_2 } }
    \\
    &\hspace{1cm}=
      \frac{ \PROBA      \Ba{ \Psi_1\np{\bm{\xi}_1,w_2} \in \Theta_1' \eqsepv
      \Phi_1\np{\bm{\xi}_1,w_2}=w_1 \eqsepv \Phi_2\np{\bm{\xi}_2,\bm{\upsilon}_1}=w_2 } }%
      { \PROBA\Ba{ \Phi_1\np{\bm{\xi}_1,w_2}=w_1 \eqsepv \Phi_2\np{\bm{\xi}_2,\bm{\upsilon}_1}=w_2 } }
      \intertext{because the random variables~$\bm{\xi}_1$ and $\bm{\xi}_2$ are independent}
    &\hspace{1cm}=
      \frac{ \PROBA      \Ba{ \Psi_1\np{\bm{\xi}_1,w_2} \in \Theta_1' \eqsepv
      \bm{\upsilon}_1=w_1 \eqsepv \bm{\upsilon}_2=w_2 } }%
      { \PROBA\Ba{ \bm{\upsilon}_1=w_1 \eqsepv \bm{\upsilon}_2=w_2 } }
      \tag{by the equality~\eqref{eq:unique_solution_proof}}
    \\
    &\hspace{1cm}=
      \PROBA\Bset{ \Psi_1\np{\bm{\xi}_1,w_2} \in \Theta_1' }
      { \bm{\upsilon}_1 = w_1 \eqsepv \bm{\upsilon}_2 = w_2 }
      \tag{by definition of the conditional probability}
    \\
    &\hspace{1cm}=
      \PROBA\Bset{ \bm{\theta}_1 \in \Theta_1' }
      { \bm{\upsilon}_1 = w_1 \eqsepv \bm{\upsilon}_2 = w_2 }
      \tag{by~\eqref{eq:six_random_variables_four_relations_theta_1}}
      \eqfinp
  \end{align*}
  Doing the same with the second term of the product, we get that 
  \begin{align*}
    &\PROBA\Bset{ \bm{\theta}_1 \in \Theta_1' \eqsepv \bm{\theta}_2 \in \Theta_2'}
      { \bm{\upsilon}_1 = w_1 \eqsepv \bm{\upsilon}_2 = w_2 }
    \\
    &\hspace{1cm}=
      \PROBA\Bset{ \bm{\theta}_1 \in \Theta_1' }
      { \bm{\upsilon}_1 = w_1 \eqsepv \bm{\upsilon}_2 = w_2 }
      \times
      \PROBA\Bset{ \bm{\theta}_2 \in \Theta_2' }
      { \bm{\upsilon}_1 = w_1 \eqsepv \bm{\upsilon}_2 = w_2 }
      \eqfinp  
  \end{align*}
  \smallskip

  This ends the proof.

\end{proof}

\subsubsection{Graphical discussion on Lemma~\ref{lem:independence}}
\label{Graphical_discussion}

A graphical representation of the system of random variables
described in Lemma~\ref{lem:independence} necessarily contains a cycle
between~$\bm{\upsilon}_1$ and $\bm{\upsilon}_2$, because
of~\eqref{eq:six_random_variables_four_relations_upsilon_1}--\eqref{eq:six_random_variables_four_relations_upsilon_2}.
As a consequence, classical results cannot be applied. 

  By contrast, using the reparametrization~\eqref{eq:unique_solution_mappings_inproof} of 
  Equations~\eqref{eq:six_random_variables_four_relations_upsilon_1}
  and~\eqref{eq:six_random_variables_four_relations_upsilon_2} --- giving 
  $\bm{\upsilon}_1= \widetilde\Phi_1 (\bm{\xi}_1 , \bm{\xi}_2)
                      \eqfinv$
   and
   $ \bm{\upsilon}_2 = \widetilde\Phi_2 (\bm{\xi}_1 , \bm{\xi}_2)$ ---
   we obtain a graphical representation which is free of cycle. However, this is at the cost of
                    loosing some properties of the initial
                    parametrization. Indeed,
system~\eqref{eq:six_random_variables} 
becomes
\begin{subequations}
  \begin{align}
    \bm{\theta}_1 &= \Psi_1 (\bm{\xi}_1 ,\bm{\upsilon}_2)
                    \eqfinv
    \\
    \bm{\theta}_2 &= \Psi_2 (\bm{\xi}_2 ,\bm{\upsilon}_1)
                    \eqfinv
    \\
    \bm{\upsilon}_1 &= \widetilde\Phi_1 (\bm{\xi}_1 , \bm{\xi}_2)
                      \eqfinv
    \\
    \bm{\upsilon}_2 &= \widetilde\Phi_2 (\bm{\xi}_1 , \bm{\xi}_2)
                      \eqfinv 
  \end{align}
  \label{eq:six_random_variables_after_transformation} 
\end{subequations}
and its DAG representation is now the one displayed in Figure~\ref{fig:lem-indep}. 
In Figure~\ref{fig:lem-indep}, we observe that there exists an unblocked path 
\( \bm{\theta}_1 \leftarrow \bm{\xi}_1 \to \bm{\upsilon}_1 \leftarrow \bm{\xi}_2 \to \bm{\theta}_2 \)
from $\bm{\theta}_1$ to $\bm{\theta}_2$.
As a consequence, we cannot conclude about the
conditional independence of $\bm{\theta}_1$ and $\bm{\theta}_2$ 
with respect to $(\bm{\upsilon}_1,\bm{\upsilon}_2)$.

By contrast, with Lemma~\ref{lem:independence} we reach the conclusion 
that the random variables~$\bm{\theta}_1$ and $\bm{\theta}_2$
are independent when conditioned on $(\bm{\upsilon}_1,\bm{\upsilon}_2)$. 

\begin{figure}[hbtp]
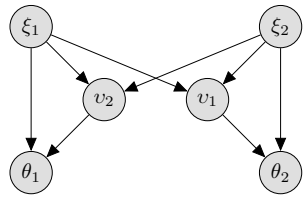

  \centering
  \resizebox{4cm}{!}{%
    \tikz{ %
      \node[obs] (X1) {$\bm{\xi}_1$} ; %
      \node[obs, below right=of X1] (W2) {$\bm{\upsilon}_2$} ; %
      \node[obs, right=of W2] (W1) {$\bm{\upsilon}_1$} ; %
      \node[obs, above right=of W1] (X2) {$\bm{\xi}_2$} ; %
      \node[obs, below right=of W1] (Y2) {$\bm{\theta}_2$} ; %
      \node[obs, below left=of W2] (Y1) {$\bm{\theta}_1$} ; %
      \edge {X1} {W2} ; %
      \edge {X1} {Y1} ; %
      \edge {X2} {W1} ; %
      \edge {X2} {Y2} ; %
      \edge {X1} {W1} ; %
      \edge {X2} {W2} ; %
      \edge {W2} {Y1} ; %
      \edge {W1} {Y2} ; %
    }
  }
  \caption{DAG representation of the
    system of equations~\eqref{eq:six_random_variables_after_transformation}}
  \label{fig:lem-indep}
  
\end{figure}

\subsection{Discrete or continuous?  It does matter}
\label{subseq:discrete-continuous}

It is notable that Lemma~\ref{lem:independence} seems to be in
contradiction with an example from~\citep{10.5555/2074158.2074214}
(recently cited in~\cite[Example~6.1]{bongers2020foundations}).

 \renewcommand{\epsilon}{R}

\begin{example}[from~\citep{10.5555/2074158.2074214}]
  Spirtes considers the following model (with $\epsilon_X$, $\epsilon_Y$,
  $\epsilon_Z$, $\epsilon_W$ being   standard independent normal random variables):
  \begin{subequations}
  \begin{align}
    X& = \epsilon_X\\
    Y& = \epsilon_Y\\
    Z&= WY +\epsilon_Z\\
    W&= Z X + \epsilon_W
  \end{align}
  \end{subequations}
  Spirtes shows that   $X$ and  $Y$ are not independent given $\np{ Z, W}$. 
  However if we set (with obvious notations related to Lemma~\ref{lem:independence})
  \begin{subequations}  
  \begin{align}
    v_1& = Z \eqfinv\\
    v_2& = W \eqfinv\\
    \xi_1&= (\epsilon_Z,Y) \eqfinv\\
    \xi_2&= (\epsilon_W,X) \eqfinv\\
    \theta_1&=\psi_1(\xi_1,v_2) = \psi_1( (\epsilon_Z,Y) ,v_2):=Y \eqfinv\\
    \theta_2&=\psi_2(\xi_2,v_1) = \psi_2( (\epsilon_W,X) ,v_1):=X \eqfinv\\
    \Phi_1(\xi_1, \bm{\upsilon}_2)& =   \Phi_1((\epsilon_Z,Y),W) := WY+\epsilon_Z \eqfinv   \\
    \Phi_2(\xi_2, \bm{\upsilon}_1)  &=\Phi_2((\epsilon_W,X),Z) := ZX+\epsilon_W
              \eqfinv
  \end{align}
    \end{subequations}
  then we see that a countable version of this example could be
  treated with Lemma~\ref{lem:independence}.
  In particular, $X$ and $Y$ are independent given $(Z,W)$, which
  is different from Spirtes's conclusion. 
\end{example}

The countable assumption in Lemma~\ref{lem:independence} seems to
draw a line between the systems described in the present paper and the
approach presented in~\citep{bongers2020foundations}.
Hence, we have an example of system for which a conditional
independence property depends on whether the codomain of $\xi_1$ and
$\xi_2$ is discrete or continuous. 
We mention that a phenomenon of the same
flavour is discussed in \cite{barbie2014topology}.

\subsection{Topological separation implies conditional independence}
\label{Topological_separation_implies_conditional_independence}

We now use the results obtained in~\S\ref{Topological_separation_implies_factorization}
and in~\S\ref{Nonrecursive_solvability_implies_conditional_independence} 
to state a general
result of conditional independence (Theorem~\ref{th:push-forward-probability-decomposition}), a corollary of which
(Theorem~\ref{th:do-calculus3rules}) constitutes a new version of Pearl's rule of do-calculus. 




\begin{theorem}
  \label{th:push-forward-probability-decomposition}
  We suppose that the assumptions of Lemma~\ref{th:decomposition} are
  satisfied.
  %
  Moreover, we suppose that the set~$\Omega$ in~\eqref{eq:states_of_Nature_product}
  is equipped with a probability~$\PROBA=\bigotimes_{\agent \in \AGENT}
  \PROBA_{\agent}$ where each $\PROBA_{\agent}$ is a probability on
  \( \np{\Omega_{\agent}, \tribu{\NatureField}_{\agent}} \).

  We define the following pushforward 
  probability~$\QQ_\wstrategy$ on~\( \np{\HISTORY,\tribu{\History}} \), in~\eqref{eq:HISTORY}, by
  %
  \begin{equation}
    \QQ_\wstrategy = \PROBA\circ \SolutionMap_\wstrategy^{-1} 
    \eqfinp 
  \end{equation}
  Then, \( \bp{\HISTORY,\tribu{\History},\QQ_\wstrategy} \) is a probability
  space, and the two projections
  \( \projection_{\TopologicalClosure{\AgentSubsetY}} : \np{\HISTORY,\tribu{\History}} 
  \to \np{\CONTROL_{\TopologicalClosure{\AgentSubsetY}},
    \tribu{\Control}_{\TopologicalClosure{\AgentSubsetY}}} \)
  and 
  \( \projection_{\TopologicalClosure{\AgentSubsetZ}} : \np{\HISTORY,\tribu{\History}} 
  \to \np{\CONTROL_{\TopologicalClosure{\AgentSubsetZ}},
    \tribu{\Control}_{\TopologicalClosure{\AgentSubsetZ}}} \) as in~\eqref{eq:projection_Bgent}
  are independent under~\( \QQ_\wstrategy \), 
  conditionally on the subset~$\HistorySubset\subset \HISTORY$ 
  and on the projection~\( \projection_{\AgentSubsetW} : \np{\HISTORY,\tribu{\History}} 
  \to \np{\CONTROL_{\AgentSubsetW},
    \tribu{\Control}_{\AgentSubsetW} } \).
\end{theorem}
Theorem~\ref{th:push-forward-probability-decomposition} claims that, for any values
$\control_{\AgentSubsetY} \in\CONTROL_{\TopologicalClosure{\AgentSubsetY}}$,
$\control_{\AgentSubsetZ} \in \CONTROL_{\TopologicalClosure{\AgentSubsetZ}}$
and $\control_{\AgentSubsetW} \in \CONTROL_{\AgentSubsetW}$,
we have that 
\begin{align}
  \QQ_\wstrategy
  & \Bsetp{
    \projection_{{\TopologicalClosure{\AgentSubsetY}}}\np{\history}
    = \control_{\AgentSubsetY} \eqsepv
    \projection_{\TopologicalClosure{\AgentSubsetZ}}\np{\history}
    = \control_{\AgentSubsetZ}
    }{
    \history\in\History \eqsepv
    \projection_{\AgentSubsetW}\np{\history} =\control_{\AgentSubsetW} }
    \nonumber \\
  &= \QQ_\wstrategy
    \bsetp{
    \projection_{{\TopologicalClosure{\AgentSubsetY}}}\np{\history}
    = \control_{\AgentSubsetY} 
    }{
    \history\in\History \eqsepv
    \projection_{\AgentSubsetW}\np{\history} =\control_{\AgentSubsetW} }
    \nonumber \\
  &\phantom{==}  \times 
    \QQ_\wstrategy\bp{
    \projection_{\TopologicalClosure{\AgentSubsetZ}}\np{\history}
    = \control_{\AgentSubsetZ}
    \mid
    \history\in\History \eqsepv
    \projection_{\AgentSubsetW}\np{\history} =\control_{\AgentSubsetW} }
    \eqfinp     
    \nonumber     
\end{align}  

\begin{proof}
  If \( \PROBA\bp{\SolutionMap_{\wstrategy}^{-1}(\HistorySubset)}=0 \),
  conditional independence is trivial (and meaningless!).
  We suppose that \( \PROBA\bp{\SolutionMap_{\wstrategy}^{-1}(\HistorySubset)}>0 \)
  and we instantiate Lemma~\ref{lem:independence} with
  \begin{itemize}
  \item 
    probability space  \(\tilde{\Omega} =
    \SolutionMap_{\wstrategy}^{-1}(\HistorySubset) \) with renormalized probability 
    \( \tilde{\PROBA}=\PROBA/\PROBA\bp{\SolutionMap_{\wstrategy}^{-1}(\HistorySubset)} \),
  \item 
    six measurable spaces 
    $\Xi_1=\Omega_{\widetilde{Y}}$, $\Xi_2=\Omega_{\widetilde{Z}}$, 
    $\Upsilon_1=\CONTROL_{\AgentSubsetW_\AgentSubsetY}$, $\Upsilon_2=\CONTROL_{\AgentSubsetW_\AgentSubsetZ}$, 
    $\Theta_1=\CONTROL_{\AgentSubsetY\cup \AgentSubsetY'}$, $\Theta_2=\CONTROL_{\AgentSubsetZ\cup \AgentSubsetZ'}$,
  \item 
    four measurable mappings 
    \( \Psi_1 =\projection_{\AgentSubsetY\cup \AgentSubsetY'} \circ\partialReducedSolutionMap_{\policy_{\AgentSubsetY}} \),
    \( \Psi_2 =\projection_{\AgentSubsetZ\cup \AgentSubsetZ'}\circ\partialReducedSolutionMap_{\policy_{\AgentSubsetZ}} \),
    \( \Phi_1 =\projection_{\AgentSubsetW_\AgentSubsetY} \circ\partialReducedSolutionMap_{\policy_{\AgentSubsetY}} \),
    \( \Phi_2=\projection_{\AgentSubsetW_\AgentSubsetZ}\circ \partialReducedSolutionMap_{\policy_{\AgentSubsetZ}} \),
  \item 
    six random variables 
    \(\bm{\xi}_1\np{\omega}  = \omega_{\widetilde{Y}} \),  
    \(\bm{\xi}_2\np{\omega}  = \omega_{\widetilde{Z}} \), 
    for all \( \omega \in \tilde{\Omega} \), 
    and 
    \( \bm{\theta}_1 = \projection_{\AgentSubsetY\cup \AgentSubsetY'}
    \circ\SolutionMap_{\policy}\),
    \( \bm{\theta}_2 = \projection_{\AgentSubsetZ\cup \AgentSubsetZ'}
    \circ\SolutionMap_{\policy}\),
    \( \bm{\upsilon}_1 = \projection_{\AgentSubsetW_\AgentSubsetY}
    \circ\SolutionMap_{\policy}\),
    \( \bm{\upsilon}_2 = \projection_{\AgentSubsetW_\AgentSubsetZ}
    \circ\SolutionMap_{\policy}\)
    on \( \tilde{\Omega} \).
  \end{itemize}
  By assumption, the set~$\Omega$ in~\eqref{eq:states_of_Nature_product}
  is equipped with a probability~$\PROBA=\bigotimes_{\agent \in \AGENT}
  \PROBA_{\agent}$ where each $\PROBA_{\agent}$ is a probability on
  \( \np{\Omega_{\agent}, \tribu{\NatureField}_{\agent}} \).
  Because of the product structure, the random variables 
  \( \bm{\xi}_1 \) and \( \bm{\xi}_2 \) are independent with respect to~$\tilde{\PROBA}$.

  As the assumptions of Lemma~\ref{th:decomposition} are satisfied,
  Equation~\eqref{eq:factorization_bis} holds true, that is, 
  we have that 
  \begin{equation*}
    \begin{split}
      \ReducedSolutionMap_\wstrategy\np{\omega} = \bgp{
        \partialReducedSolutionMap_{\wstrategy_{\widetilde{\AgentSubsetY}}}
        \Bp{ \omega_{\widetilde{\AgentSubsetY}}, 
          \projection_{\AgentSubsetW_\AgentSubsetZ}\bp{\SolutionMap_{\wstrategy}\np{\omega}} },
        \partialReducedSolutionMap_{\wstrategy_{\widetilde{\AgentSubsetZ}}}
        \Bp{ \omega_{\widetilde{\AgentSubsetZ}}, 
          \projection_{\AgentSubsetW_\AgentSubsetY}\bp{\SolutionMap_{\wstrategy}\np{\omega}} },
        \partialReducedSolutionMap_{\wstrategy_{\residual}}
        \Bp{ \omega_{\residual}, 
          \projection_{\widetilde{\AgentSubsetY}\cup\widetilde{\AgentSubsetZ}}\bp{\SolutionMap_{\wstrategy}\np{\omega}} }
      } 
      \eqsepv
      \\ \forall \omega \in \SolutionMap_{\wstrategy}^{-1}(\HistorySubset) 
      \eqfinp 
    \end{split}
  \end{equation*}
  Thus, the assumptions of Lemma~\ref{lem:independence} are satisfied,
  and we conclude that the random variables 
  $\bm{\theta}_1$ and $\bm{\theta}_2$ are independent under the probability~$\tilde{\PROBA}$,
  when conditioned on $(\bm{\upsilon}_1, \bm{\upsilon}_2)$.

  In other words, we have obtained that 
  $\projection_{\AgentSubsetY\cup
    \AgentSubsetY'}\circ\SolutionMap_{\policy}= \projection_{\bar{
      \AgentSubsetY}}\circ\SolutionMap_{\policy}$ and
  $\projection_{\AgentSubsetZ\cup
    \AgentSubsetZ'}\circ\SolutionMap_{\policy}=\projection_{\bar{
      \AgentSubsetZ}}\circ\SolutionMap_{\policy}$ 
  are independent random variables, when conditioned on
  \(\projection_{\AgentSubsetW_\AgentSubsetY}
  \circ\SolutionMap_{\policy}\)
  and
  \(\projection_{\AgentSubsetW_\AgentSubsetZ}
  \circ\SolutionMap_{\policy}\) 
  under the probability~$\tilde{\PROBA}$.
  We deduce that 
  $ \projection_{\bar{
      \AgentSubsetY}}$ and
  $\projection_{\bar{
      \AgentSubsetZ}}$ are independent
  when conditioned on
  \(\projection_{\AgentSubsetW_\AgentSubsetY}\)
  and
  \(\projection_{\AgentSubsetW_\AgentSubsetZ}\)
  under the probability~\( \QQ_\wstrategy = \PROBA\circ \SolutionMap_\wstrategy^{-1} \).

  This ends the proof.
\end{proof}

\subsection{Topological separation implies the do-calculus}
\label{Topological_separation_implies_the_do-calculus}

Next we deduce from Theorem~\ref{th:push-forward-probability-decomposition}
a variant of Pearl's do-calculus. 


\begin{theorem}[Do-calculus in W-models]
  \label{th:do-calculus3rules}
  Under the assumptions of Theorem~\ref{th:push-forward-probability-decomposition},
  the projection~\( \projection_{\AgentSubsetY} \) 
  has the same conditional distribution under~\( \QQ_\wstrategy \), 
  whether the conditioning is \wrt\
  the subset~$\HistorySubset\subset \HISTORY$,
  the projection~\( \projection_{\AgentSubsetW} \)
  and the projection~\( \projection_{\TopologicalClosure{\AgentSubsetZ}} \),
  or is only \wrt\
  the subset~$\HistorySubset\subset \HISTORY$ and 
  the projection~\( \projection_{\AgentSubsetW} \).
\end{theorem}

\begin{proof}
  By
  Theorem~\ref{th:push-forward-probability-decomposition}, when 
  conditioning with respect to $\History$ and \( \projection_{\AgentSubsetW} \) , \( \projection_{\AgentSubsetY} \) and \(
  \projection_{\TopologicalClosure{\AgentSubsetZ}} \) are
  independent. This implies
  (see for example
  \cite[Proposition~2.4~(c)]{vanPutten-vanSchuppen:1985}) in
  particular that  \(\projection_{\TopologicalClosure{\AgentSubsetZ}} \)
  can be removed from the conditioning above mentioned.
\end{proof}

We have proved, loosely speaking, that
\begin{equation}
  \AgentSubsetY \ConditionalTopologicalSeparation \AgentSubsetZ \mid \np{\AgentSubsetW,\HistorySubset} \implies
  \QQ_{\policy}(\history_\AgentSubsetY
  |\history_\AgentSubsetW,\history_{\TopologicalClosure{\AgentSubsetZ}},\HistorySubset)
  = \QQ_{\policy}(\history_\AgentSubsetY |\history_\AgentSubsetW,\HistorySubset) 
  \eqfinp 
\end{equation}
In particular
\begin{equation}
  \label{eq:do-calculus}
  \AgentSubsetY \ConditionalTopologicalSeparation \AgentSubsetZ \mid \np{\AgentSubsetW,\HistorySubset} \implies
  \QQ_{\policy}(\history_\AgentSubsetY
  |\history_\AgentSubsetW,\history_{\AgentSubsetZ},\HistorySubset)
  = \QQ_{\policy}(\history_\AgentSubsetY |\history_\AgentSubsetW,\HistorySubset) 
  \eqfinp 
\end{equation}
We stress the conciseness of Theorem~\ref{th:do-calculus3rules} ---
permitted by the notions  introduced in this paper ---
as we now show that it implies the three rules of Pearl,
as well as the following two recent results.
As already mentioned in Example~\ref{example:tikka},
the authors in~\citep{tikka2019identifying} manage to summarize the
three rules of do-calculus thanks to the notion of context specific
independence.
They rely on so-called \emph{labeled DAG} that can be turned into a context specific DAG by removing
the arcs that are desactivated (spurious) in the context of interest.
In the formalism that we propose, such context is represented by a subset of~$\HISTORY$.
Indeed, if we denote by $\History\in\HISTORY$ the context for which an arc~\(
\np{\agent,\bgent} \) is deactivated (in the language of
\citep{tikka2019identifying}),
we represent this by the following two properties: 
\( \agent \not\in\Precedence_{\emptyset,\History}\bgent \),
\( \agent \in\Precedence_{\emptyset,\History^c}\bgent \).
%
Such a property can be also be encoded in the information set of agent~$\bgent$.
As a consequence, there is a mapping from the model introduced in
\citep{tikka2019identifying} to W-models.

To introduce the next result, 
we will allow some abuse of notations to make our notations as close as possible to the
literature we are comparing with.   
We will use, for $\Bgent\subset \AGENT$ and $\control_\Bgent\in
\CONTROL_{\Bgent}$, the notation
\( [\history_{\Bgent} = \control_{\Bgent}] =
\nset{\history\in\HISTORY}{\history_{\Bgent} = \control_\Bgent} \).
%
Then, Rule~1 in \cite{tikka2019identifying}  rewrites, in our setting, as 
\begin{equation}
  \label{rule1Tikka}
  Y \ConditionalTopologicalSeparation_{(X,\history_{\tilde{X}} = \control_{\tilde{X}})} Z \implies
  \QQ\left(\history_Y | \history_Z, h_{X},\history_{\tilde{X}} =
    \control_{\tilde{X}}\right)=
  \QQ \left( \history_Y |  \history_{X},\history_{\tilde{X}} =  \control_{\tilde{X}}\right) 
\end{equation}
where $X,\tilde{X}\subset \AGENT$ and for a given value~$\control_{\tilde{x}}$.

\begin{proposition}
  \label{prop:subsuming}
  Rule~1 from \citep{tikka2019identifying} can be deduced from Theorem
  \ref{th:do-calculus3rules}.
  In particular, Theorem~\ref{th:do-calculus3rules}  subsumes Pearl's  do-calculus
  from \cite{pearl1995causal}.
\end{proposition}

\begin{proof}
  If we set $W = X$ and $\HistorySubset =\{\history\in H; \history_{\tilde{X}} =
  \control_{\tilde{X}}\} $ in Equation~\eqref{eq:do-calculus}
  (obtained with  Theorem \ref{th:do-calculus3rules})
  we obtain~\eqref{rule1Tikka} which is Rule~1 from \cite{tikka2019identifying}.
 The proof of  Theorem 2 from \cite{tikka2019identifying} states that this rule
implies in particular the rules of Pearl's do-calculus.
  \smallskip
\end{proof}

\section{Discussion}

In this paper, we simplify and generalize the do-calculus
by leveraging the concept of information field,
using Witsenhausen's intrinsic model. 
The do-calculus is reduced to one rule. 
We underline that the results are consequences of the information structure, but
have nothing to do with the probability.
For most cases, one only needs to understand the notion of inverse
image  to work with information fields on top of SCMs and DAGs.
In exchange, information fields provide a compact, unifying and
versatile language that brings new intuitions on the causal structure
of the problem.

For instance, we have illustrated why the notion of topological separation is practical: once the splitting of the
conditioning variables  known, checking that  
an intersection  is empty  is easier than checking  a blocking condition on a collection of paths.
We prove in~\citep{De-Lara-Chancelier-Heymann-2021} that the topological separation is equivalent to the d-separation on DAGs.

The Information Dependency Model is a good candidate to bring
uniformity and consistency in lieu of \emph{ad hoc} frameworks.
It can be a temporary detour to introduce
new notions, for instance the Definition~\ref{de:conditional_precedence_relation} of conditional precedence
would have been harder to express with the SCM as primitive. 

In addition, we have presented and solved  an  example that cannot be handled easily with the
current state of the literature.  

Last, we mention that the notion of well-posedness we use was
introduced in~\citep{Witsenhausen:1975} half a century ago for another
field of applied mathematics. It is
interesting to observe that this notion could serve a new purpose in
the field of causal inference. 

Further work includes drawing connections with other research programs,
such as questions related to identification
causal structure
\citep{shpitser2006identification,shpitser2008complete,tikka2019causal}
or extensions of do-calculus~\citep{correa2020a}. 
As argued in
Sect.~\ref{subseq:discrete-continuous}, there is a fundamental
difference between the discrete and continuous case that calls for
different tooling; in this regard, it would be
interesting to study  the connections of this work
with~\cite{bongers2020foundations,pmlr-v115-forre20a}.

\section*{Aknowledgements}

We thank the organizers of
the Causal Discovery and Causality-Inspired Machine Learning Workshop at Neural
Information Processing Systems,
where we could present this work on 11 December 2020.
We thank Sridhar Mahadevan for interesting exchanges. 


\newcommand{\noopsort}[1]{} \ifx\undefined\allcaps\def\allcaps#1{#1}\fi

\end{document}